\PassOptionsToPackage{table,xcdraw}{xcolor}
\documentclass[sigconf]{acmart}
\setcopyright{none}

\acmConference[]{This is the author's version of the work}{the definitive version of record was published in CCS '22}{and is accessible as follows.}
\acmDOI{10.1145/3548606.3560582}
\acmISBN{978-1-4503-9450-5/22/11}

\usepackage{subcaption}
\usepackage{amsmath,amsfonts}
\usepackage{algorithmic}
\usepackage{graphicx}
\usepackage{textcomp}
\usepackage[linesnumbered,ruled,vlined]{algorithm2e}
\usepackage{hyperref}
\usepackage{enumitem}
\usepackage{tabularx}
\usepackage{booktabs}
\usepackage{amsthm}
\usepackage{verbatim}
\usepackage{array,booktabs}
\usepackage{mathtools}
\usepackage{multirow}

\usepackage[binary-units,per-mode=symbol]{siunitx}
\usepackage[capitalize,noabbrev]{cleveref}

\newcommand{\system}{Helia}
\newcommand{\dataplane}{data plane}

\crefname{section}{\S}{\S}
\crefname{page}{page}{pages}
\crefname{paragraph}{Section}{Sections}
\Crefname{section}{Section}{Sections}
\crefformat{footnote}{#2\footnotemark[#1]#3}
\crefname{equation}{Eq.}{Eqs.}
\Crefname{equation}{Equation}{Equations}

\newcommand{\prf}[2]{\ensuremath{\textsc{PRF}_{#1}(#2)}}
\newcommand{\mac}[2]{\ensuremath{\textsc{MAC}_{#1}(#2)}}

\newcommand{\authenc}[2]{\ensuremath{\textsc{AuthEnc}_{#1}(#2)}}
\newcommand{\key}[1]{\ensuremath{\textsc{K}_{#1}}}

\newcommand{\auth}[1]{\ensuremath{\alpha_{#1}}}
\newcommand{\bw}[1]{\ensuremath{\beta_{#1}}}
\newcommand{\rvf}[1]{\ensuremath{\varphi_{#1}}}
\newcommand{\tsexp}[1]{\ensuremath{\mathrm{tsExp}_{#1}}}
\newcommand{\tspkt}[1]{\ensuremath{\mathrm{tsPkt}}}
\newcommand{\lenb}[1]{\ensuremath{\mathrm{lenB}}}
\newcommand{\lenvf}[1]{\ensuremath{\mathrm{lenVF}}}
\newcommand{\lenpkt}[1]{\ensuremath{\mathrm{len(pkt)}}}
\newcommand{\matentry}[3]{\ensuremath{\textsc{M}_{(#1)}^{#2,#3}}}
\newcommand{\drkey}[2]{\ensuremath{\textsc{K}_{#1 \to #2}}}
\newcommand{\bloom}[1]{\ensuremath{\textsc{B}_\text{#1}}}

\allowdisplaybreaks

\SetKwInput{KwInput}{Input}                % Set the Input
\SetKwInput{KwOutput}{Output}              % set the Output

\newcommand{\enumfont}{\bfseries}
\SetEnumitemKey{trait}{label={\enumfont T\arabic*}, ref=T\arabic*}

\newcommand{\nint}[1]{\textrm{round}(#1)}

\begin{document}

\title{Protecting Critical Inter-Domain Communication through Flyover Reservations}

% How many authors per line
\settopmatter{authorsperrow=4}

\author{Marc Wyss}
\affiliation{%
  \institution{ETH Zurich}
  %\country{Switzerland}}
  }
\email{marc.wyss@inf.ethz.ch}

\author{Giacomo Giuliari}
\affiliation{%
  \institution{ETH Zurich}
  %\country{Switzerland}}
  }
\email{giacomog@inf.ethz.ch}

\author{Jonas Mohler}
\affiliation{%
  \institution{ETH Zurich}
  %\country{Switzerland}}
  }
\email{jonas@mohlers.ch}

\author{Adrian Perrig}
\affiliation{%
  \institution{ETH Zurich}
  %\country{Switzerland}
}
\email{adrian.perrig@inf.ethz.ch}

\begin{abstract}
    To protect against naturally occurring or adversely induced congestion in the
    Internet, we propose the concept of \emph{flyover reservations}, a
    fundamentally new approach for addressing the availability demands of
    critical low-volume applications. In contrast to path-based reservation
    systems, flyovers are fine-grained ``hop-based'' bandwidth reservations on
    the level of individual autonomous systems. We demonstrate the scalability
    of this approach experimentally through simulations on large graphs.
    Moreover, we bring the flyovers' potential to full fruition by introducing
    \emph{\system{}}, a protocol for secure flyover reservation setup and data
    transmission. We evaluate \system{}'s performance based on an
    implementation in DPDK, demonstrating authentication and forwarding of
    reservation traffic at \SI{160}{Gbps}. Our security analysis shows that
    \system{} can resist a large variety of powerful attacks against
    reservation admission and traffic forwarding.
    Despite its simplicity, \system{} outperforms current state-of-the-art
    reservation systems in many key metrics.
\end{abstract}

\maketitle
\pagestyle{plain}

\section{Introduction}

Given the lack of delivery guarantees for traffic traversing the Internet,
companies requiring high availability are forced to turn to other more
expensive networked services. ISPs provide high-uptime connectivity services to
their customer's critical applications---financial services, command and
control, and others---for a hefty premium. 
These advanced services take the form of end-to-end \emph{bandwidth
reservations}, e.g., MPLS tunnels, where a certain amount of bandwidth is
exclusively allocated for the communication between the endpoints. Traffic is
then shielded from external congestion even in the case of denial of service
(DoS) attacks. However, such solutions are inflexible and expensive, as they
require a single entity to manage the whole infrastructure between the
endpoints~\cite{QoSbook}.
In an effort to achieve the benefits of bandwidth reservations at lower cost,
and for a larger fraction of traffic, recent work proposed systems that enable
\emph{inter-domain} bandwidth reservations~\cite{Colibri,GLWP}. The protocols
thus developed reserve bandwidth for individual traffic sources across
autonomous systems (ASes, the networks forming the Internet), on the whole
communication path to the destinations.
%These reservations are cryptographically authenticated, and each reservation
%packet is forwarded with priority, protecting it from congestion.
Reservation traffic is cryptographically authenticated and forwarded with
priority, protecting it from congestion.

\medskip\noindent
In this paper,  we present a radically new design for inter-domain bandwidth
reservations. Instead of reserving bandwidth on an entire path as in previous
systems, sources can reserve bandwidth for single AS hops. Sources can compose these ``hop'' reservations---which we call \emph{flyovers}
\footnote{Analogous to highway overpasses alleviating congestion at
intersections.}
---to
create end-to-end--protected paths across the Internet. 
This simple construction has two major consequences. First, it simplifies the
reservation admission and accounting. Flyover reservations are two-party
contracts between the source and the remote AS (the \emph{flyover provider}),
and can thus avoid the complex setup procedure that, in previous systems,
requires the active involvement of all on-path ASes. Second, flyovers allow a
more effective use of resources, since sources have freedom in the allocation
of flyover bandwidth to flows. In a path-based approach, the reservation is
tied to the whole path, and bandwidth is wasted if the flow terminates before
the reservation expires. With flyovers, sources can allocate bandwidth to
concurrent flows sharing parts of a path, reducing the waste.
\emph{In essence, the simplicity and flexibility of flyovers enables more
efficient bandwidth reservation protocols, that are then also more impervious
to resource-exhaustion DoS attacks.}

We back this claim by designing and implementing Helia, a fully fledged
inter-domain bandwidth reservation system.
Thanks to its flyover-based design, Helia can achieve over \num{2}$\times$ faster forwarding of reservation-protected traffic, more than \num{10000}$\times$ larger reservations in the median, and up to \num{4} orders of magnitude faster reservation bandwidth computation w.r.t. the state-of-the-art systems we compared against.
Helia allows traffic sources to securely create flyover reservations with
remote ASes, compose them into an end-to-end path, and forward critical traffic
under the protection of the reservation at record speeds: with \SI{6.45}{Mpps} on a
single core, it achieves \SI{150}{Gbps} of flyover traffic using \num{4} cores on a commodity server. Even more, the
reservation setup protocol is so simple it can be run on the data fast path,
removing the need for additional control-plane infrastructure.

In line with previous work, we secure Helia against attacks. Thanks to strong
per-packet source authentication, the system can detect reservation forgery,
spoofing, and the overuse of legitimate reservations. However, the
flyover-based design also provides additional benefits from a security and
availability standpoint. Since flyovers can be reused to compose multiple
paths, the number of reservation setup rounds---the most vulnerable step in the
protocol---is lower than in path-based reservation protocols, reducing the
attack surface. Then, the vast speedup in the admission procedure makes DoS
attacks on this component much harder.
Further, overuse monitoring is much more timely and precise in Helia: Since
flyover providers only have to keep track of reservation traffic sources, they
can deploy deterministic monitoring of reservation overuse, with zero error
and minimal overhead---only \SI{800}{kB} of memory for one hundred thousand
ASes.\footnote{There are around \num{75} thousand ASes in the Internet
today~\cite{cidr}.} The potentially exponential number of paths, on the
contrary, forces path-based systems to use probabilistic monitoring schemes,
which are not fully reliable in detecting misbehavior and require separate
infrastructure.
Finally, the simple reservation accounting at flyover providers allows us to
prove a bounded ``time to reservation'': within this deadline from the first
request packet a reservation will be made available to the source. This is
particularly useful when trying to establish a reservation under a DDoS attack.
The price that \system{} pays for the simplicity of its reservation accounting is low reservation granularity; the source can not negotiate the size of its flyover reservations, but instead obtains a fair share of the total flyover bandwidth.

\medskip\noindent
In summary, we develop \textbf{two major contributions} to advance the
state-of-the-art of secure bandwidth reservation systems.
We first \textbf{establish the algorithmics behind flyover reservations},
specifying how to compose flyovers, and designing an algorithm to assign
flyover bandwidth fairly. We show through simulations on large random graphs
that flyover composition provides superior scalability compared to
GLWP~\cite{GLWP}, a path-based bandwidth reservations system.
Second, we \textbf{design and implement \system{}, a protocol to establish and
authenticate flyover reservations at line rate and entirely on the data fast
path}. This surprising result is possible thanks to the extremely simple
algorithms developed in our first contribution and the use of high-speed
symmetric key cryptography~\cite{PISKES}.
We implement \system{} in DPDK, and benchmark its admission and forwarding capability. We finally discuss \system{}'s security and availability properties. 

Beyond our results on the scalability of flyover reservations, we hope our
paper will re-kindle the interest of the security community in the difficult
problem of providing affordable and robust traffic delivery guarantees in the
Internet.

\section{Background}
Given the broad problem setting, we provide background on inter- and intra-domain traffic engineering, reservation systems, and the building blocks of our designs.

\subsection{Traffic Engineering \& Critical Applications}

\paragraph{Autonomous Systems} 
The Internet is an interconnection of centrally operated networks called
autonomous systems (ASes), such as Internet service providers (ISPs) and
transit providers. ASes connect with each other at \emph{peering points} or
\emph{interfaces}, where border routers forward traffic between the networks.
It is therefore common to distinguish between \emph{intra-domain} forwarding,
happening inside a single AS, and \emph{inter-domain} forwarding, where traffic
is relayed across multiple ASes.
This structure is shown in \cref{fig_AS}a).

\paragraph{Critical Applications}
The majority of ASes' bandwidth is devoted to forwarding best effort traffic,
usually generated by low-priority applications such as video streaming and web
browsing. A small fraction of the deployed bandwidth is however dedicated to
premium services for high-paying customers that require extremely high uptimes. 
These critical applications include financial settlements, sensitive
information transfer, remote command and control, key-exchange protocols, and
others. One common aspect of these applications, aside from their high-reliability requirements, is their low bandwidth footprint. As an example, the 
SWIFT financial network, which accounts for a large fraction of global 
inter-bank transactions, requires on average less than \SI{1}{Mbps} between all of its \num{11000} member institutions~\cite{GLWP}.

\paragraph{Intra-Domain Reservations}
Network operators employ many diverse strategies to protect critical application
traffic from congestion and failures.
Complete traffic separation is achieved by using leased lines or MPLS tunnels
(sometimes called ``virtual leased lines'')~\cite{MPLS}, thus providing the
strongest possible guarantees. Bandwidth reservations can also be implemented
by marking packets (e.g., using DiffServ~\cite{DiffServ}), and then configuring
appropriate queuing disciplines at all intermediate switching elements to
prioritize their forwarding.
However, these systems either require ad-hoc infrastructure (leased lines), or
extensive configuration (MPLS tunnels, DiffServ), and can therefore be provided
only within an operator's domain. To this day, extending bandwidth
reservations, e.g., MPLS tunnels, across multiple ASes is a manual process,
enforced by long-standing contracts between operators. A further limiting
factor is that the security of these systems depends on their centralized
operation (e.g., to ensure that only allowed routers can set packet
flags). This assumption breaks down in the decentralized inter-domain setting.

In this paper we abstract away from the details of intra-domain reservations, and
provide a system to quickly and securely create, compose, and tear down
reservations across multiple domains.

\subsection{Secure Bandwidth Reservation Systems} \label{sec_background_reservation_systems}
A few relatively recent systems promise to overcome the security limitations of traditional bandwidth reservations protocols. To the best of our knowledge, these are: SIBRA~\cite{Basescu2016SIBRA}, GLWP~\cite{GLWP}, and Colibri~\cite{Colibri}. As SIBRA is the predecessor of Colibri, we do not directly compare against it and focus on Colibri and GLWP only.

GLWP and Colibri are both \emph{path-based} reservation protocols. In a setup
phase, the source AS sends a reservation packet over the path for which it
wants to establish a reservation, where each on-path AS adds local information 
to the request packet. Informed by this data-collection step, ASes can compute
the amount of reserved bandwidth they wish to grant to the source, and issue an
authorization token. 
This token is used by the source AS at forwarding time to
prove that its traffic is allowed to use the reservation.
Thanks to this setup phase, GLWP and Colibri ensure that no over-allocation
occurs: at any point in time, the sum of bandwidth reservations going through
some link must not be greater than the link capacity. 

GLWP and Colibri also rely on other subsystems to prevent DoS attacks: A
duplicate-suppression system filters out replayed packets~\cite{replay2017},
and a probabilistic bandwidth monitor detects the bandwidth overuse from
malicious sources~\cite{Wu2014,Sivaraman2017,loft}.
\Cref{sec_related} further describes the distinctive
properties of these two systems.

\subsection{Technical Building Blocks}

To complete the background, we describe \emph{allocation matrices}, which
encode the information required to bootstrap flyover reservations, and  DRKey~\cite{PISKES}, a symmetric key distribution system that allows Helia to scalably authenticate reservations.

\begin{figure}[t]
    \centering
    \includegraphics[width=0.4\textwidth]{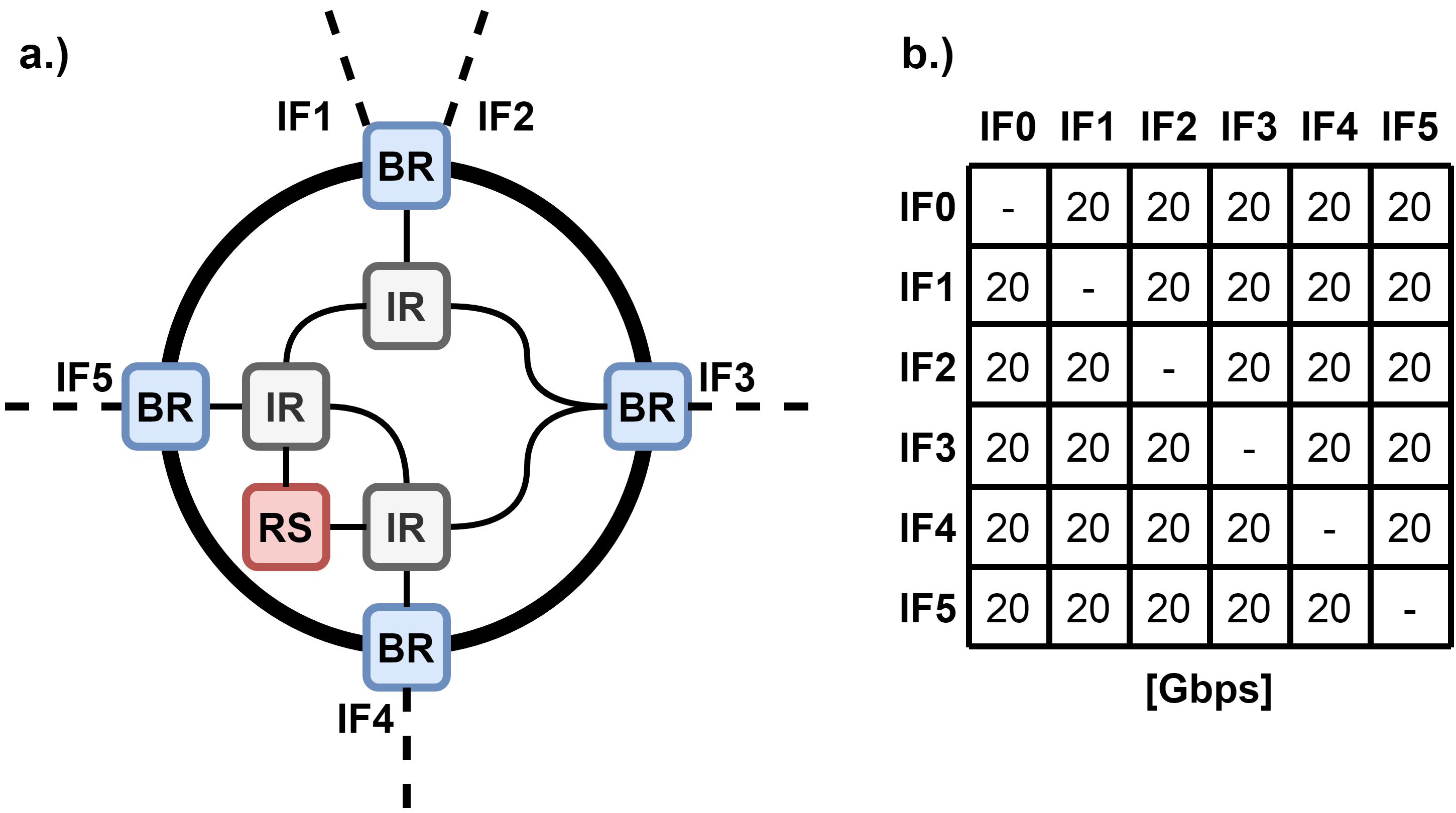}
    \caption{{\textbf{a.) Example intra-AS network.} It consists of multiple internal routers (IRs) and infrastructure components such as a reservation service (RS). Communication with other ASes happens through border routers (BRs). A border router handles one or more interfaces (IFs), where each interface is connected through a link to another AS.
    \textbf{b.) Example allocation matrix}, assuming inter-domain link capacities of \SI{100}{Gbps}. IF0 is an auxiliary interface representing communication from and to entities within the AS.
}}
    \label{fig_AS}
\end{figure}

\paragraph{Allocation Matrices}
A common way to represent intra-AS forwarding capabilities is to encode the
traffic entering and exiting the network in a (possibly time-varying) traffic
matrix.
Similarly, an \emph{allocation matrix} represents the intra-AS reservation
capacity. The allocation matrix entry \matentry{A}{i}{j} at the $i$th row and
$j$th column indicates the amount of reservable bandwidth AS~$A$ can guarantee
from interface~$i$ to interface~$j$ through its internal network.
To ensure that no congestion occurs, our reservation protocol must enforce that
the aggregate traffic traversing any interface pair $(i, j)$ is always  lower
than the maximum available reservation bandwidth \matentry{A}{i}{j}.
This property is called \emph{no over-allocation}~\cite{GMA}.
The allocation matrix does not capture the bandwidth used by
best-effort traffic. However, unused reservation
bandwidth is repurposed to also carry best-effort traffic, thus avoiding wasted resources.

\paragraph{Key Distribution with DRKey} 
Using public-key cryptography is too computationally expensive to authenticate reservation information, and it would introduce a DoS attack vector if directly applied to the problem of packet authentication.
In contrast, symmetric cryptography is efficient, but requires distributing and storing shared keys.
DRKey eliminates the need for keeping any state or fetching keys by enabling the \emph{dynamic re-computation} of keys at border routers~\cite{PISKES}.
Using DRKey, a border router in AS~$A$ can derive a symmetric key with any other AS~$B$ by applying a pseudorandom function to the identifier of AS~$B$:

\begin{equation} 
    \drkey{A}{B} = \prf{\key{A}}{B} \label{eq_drkey} \\
\end{equation}

Then, the only information the border router needs to store is \key{A}, a
secret key only known to trusted infrastructure, i.e., other border routers and
machines inside AS~$A$. On the other hand, a service in AS~$B$ needs to fetch
\drkey{A}{B} from AS~$A$, where this exchange is protected by asymmetric
cryptography. The arrow in \drkey{A}{B} indicates the asymmetry in the key
establishment: entities in AS~$A$ can derive the key on the fly (on the order
of nanoseconds), while entities in AS~$B$ have to request the key (on the order
of milliseconds).

While there exist other systems for symmetric key distribution, they might require additional router state~\cite{Passport} or a trusted authority~\cite{Kerberos}.
In this work, we assume that the DRKey system is used to establish shared keys between ASes.

\section{flyover Reservations}

In this section we present the ideas behind flyover reservations.

\subsection{Definition} \label{sec_flyovers_definition}
In our terminology, a path is composed of multiple hops, where a hop denotes the intra-domain forwarding path between the ingress and egress interface of an AS.
The main observation behind the design of flyover reservations is that bandwidth reservations do not necessarily need to be \emph{path-based}, but that they can instead be \emph{hop-based}.
Thanks to this decomposition of a path into single hops, an on-path AS can perform the admission of a bandwidth reservation for the source AS based solely on local information.
More specifically, the $i$th on-path AS provides the source AS with a reservation starting at its ingress interface $a$ and ending at its egress interface $b$, where the size of the reservation $\beta_i^{~a,b}$ is computed as
\begin{equation}
    %\beta^i_{a,b} = \frac{\text{M}_{a,b}^{i}}{\rho^{i}_{a,b}}, \label{eq_bandwidth}
    \beta_i^{~a,b} = \frac{\text{M}^{a,b}_{i}}{\text{max}(\rho_{i}^{a,b},~\rho_{\text{min}})}. \label{eq_bandwidth}
\end{equation}
Here, $\text{M}^{a,b}_{i}$ is the allocation matrix entry of AS $i$ corresponding to the interface pair $(a, b)$, and $\rho_{i}^{a,b}$ is the number of ASes requesting a reservation over that interface pair.
The constant $\rho_{\text{min}}$ prevents the provisioning of excessive amounts of bandwidth in scenarios where not many ASes request a reservation.
Packets forwarded under such a reservation can use this guaranteed share of bandwidth and are protected from congestion.
By composing flyovers along the chosen path, a source AS can construct a complete end-to-end reservation.
A discussion of alternative (but insufficient) approaches to per-interface-pair reservations can be found in \cref{sec_appendix_alternatives_to_pip_reservations}.
We investigate a demand-aware flyover design---whereby source ASes can request a minimum flyover bandwidth---in \cref{sec_appendix_demand_aware}.

\subsection{Computing \texorpdfstring{$\boldsymbol{\rho}$}{Rho}} \label{sec_finding_rho} The number $\rho$  of ASes
requesting reservations through an AS has a significant influence on the
amount of flyover reservation bandwidth (\cref{eq_bandwidth}). Moreover, a
precise accounting of the number of remote ASes requesting a reservation is
essential to the efficiency and security of flyovers: While an over-estimation
of $\rho$ may lead to an inefficient bandwidth allocation, underestimating $\rho$
can cause the over-allocation of bandwidth, jeopardizing the forwarding
guarantees of the flyover reservations.
This accounting is challenging as the number of participating ASes can
dynamically change, and we do not want to introduce additional
coordination overhead among ASes, nor rely on a globally trusted registry.

We present an algorithm with which border routers can compute a precise estimate of $\rho$ in \cref{sec_appendix_fine_grained}.
Interestingly, we prove in \cref{sec_appendix_proofs} that our algorithm can
always grant a reservation to any requesting AS \emph{within a bounded time
interval} following its first request, without ever causing over-allocation.

\subsection{Composing Flyovers} \label{sec_traffic_control}
A source AS needs to assemble multiple flyover reservations in order to achieve
end-to-end guarantees on an Internet path. This is in contrast to path-based
reservation systems, where the source has a single bandwidth reservation for
each path. 
The \emph{maximum rate} at which the source can send traffic then corresponds
to the minimum reservation bandwidth of all the flyover reservations on the
path. However, this is only the case if reservation paths do not overlap,
otherwise the common flyover reservations need to be shared among the
overlapping paths (see \cref{fig_concurrent}). As an alternative to simply
splitting a flyover reservation among multiple paths, the source AS can use
time-division multiplexing and assign each path an interval during which this
path can fully use the reservation.

The ideal assignment of reservation bandwidth to path depends on the critical
applications running at the source AS, and we therefore do not specify a
concrete algorithm for this problem. In case of multiple critical applications
running inside the source AS, it is the responsibility of the source AS to
share the reservation bandwidth among them. The whole logic for reservation
management at the source AS is abstracted in a dedicated reservation service.

\begin{figure}[t]
    \centering
    \includegraphics[width=0.3\textwidth]{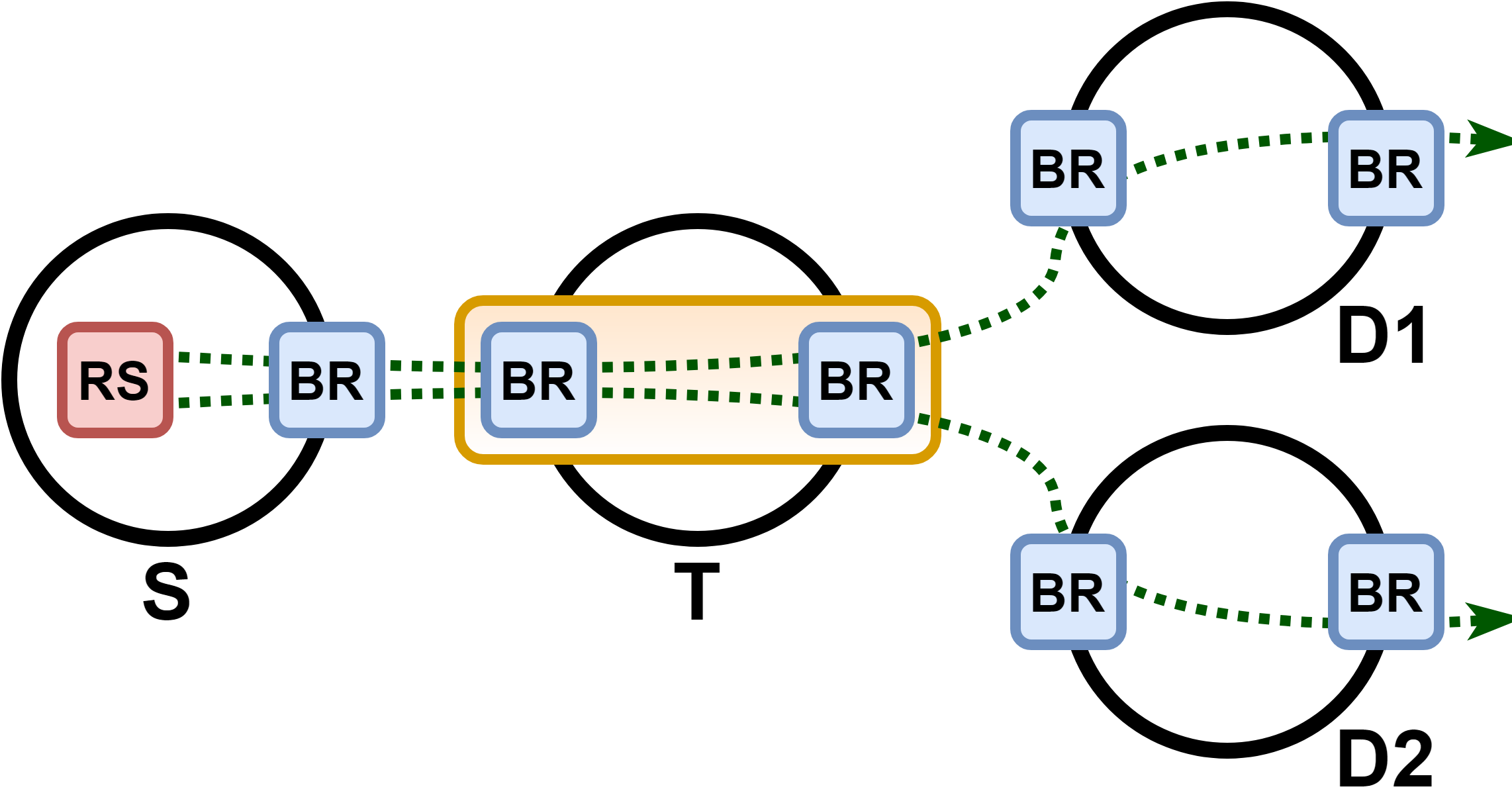}
    \caption{Scenario where traffic control is required by the source AS S. The flyover (orange rectangle) of transit AS T needs to be shared among two traffic paths (green dotted lines) traversing D1 and D2, respectively. A reservation service (RS) inside the source AS manages the reservations.
    %\todo{Remove if not enough space?}
    %Notation as in \cref{fig_AS}.
    }
    \label{fig_concurrent}
\end{figure}

\section{\system{}} \label{sec_msystem}

\subsection{Overview}
To securely instantiate flyovers, we propose \system{}.
We consider a source AS~$S$ requiring high-availability guarantees on a forwarding path composed of at least one other AS.
With \system{}, the source AS can establish flyovers with the on-path ASes, and compose them into a uni- or bidirectional end-to-end bandwidth reservation. 

In the flyover \emph{setup phase}, the source AS issues an authenticated best-effort request, which triggers on-path ASes to compute and allocate the amount of flyover bandwidth available to the source.
At on-path ASes, this admission procedure results in (i) the allocated flyover bandwidth \bw{}, (ii) the  flyover reservation expiration time \tsexp{}, and (iii) a reservation authenticator \auth{}.
The authenticator \auth{} is a cryptographic token representing an authorization for the source AS to send traffic at under a rate of \bw{} immediately until time \tsexp{} over the specific on-path AS' interface-pair.
This triplet is authenticated by the on-path AS and returned to the source. The \auth{} is also encrypted for confidentiality, as anyone holding \auth{} is able to forward reservation traffic on behalf of the source AS~$S$.
As a last step, the on-path AS registers the flyover reservation for future policing at a deterministic traffic-monitoring system, which ensures that the source AS cannot overuse its assigned rate.

In the \emph{data transmission phase}, the source AS uses the authenticators received in the setup phase to compute per-packet cryptographic tokens, called reservation validation fields (RVFs) and backward reservation validation fields (BVFs). These are included in the header of data packets at forwarding time.
RVFs and BVFs are used to efficiently authenticate the source and the reservation to the on-path ASes:
Border routers can verify them without having to store any reservation-specific information, as they can be recomputed ``on the fly'', only based on the information in the packet and an AS-local secret.
Packets that pass this validation check are then fed into the duplicate-suppression system, and are policed by the deterministic traffic monitor.
Successfully validated reservation packets are scheduled to be forwarded with highest priority, which provides increased quality of service.
To renew the reservation, the source AS can send reservation requests over the already established reservation instead of over best-effort, which guarantees successful reservation renewal. 

Henceforth, we write~$\beta_i$ instead of~$\beta_i^{~a,b}$, omitting the interfaces.

\subsection{Network Model} \label{sec_network_model}
\paragraph{Key Distribution}
We require that the source AS~$S$ has already fetched symmetric keys using DRKey from all ASes on a desired communication path, i.e., it is in possession of \drkey{i}{S} for each on-path AS~$i$. This can always be done ahead of time as keys are refreshed only daily.

\paragraph{Path Characteristics}
We assume that paths are stable over time, meaning that paths do not unexpectedly change during the lifetime of a reservation. This is crucial, as otherwise no forwarding guarantees can be achieved. 
Additionally, to support bidirectional reservations---where traffic from destination to source is also protected---we assume that paths are symmetric. 
Bidirectional reservations are only possible if the backward path is the reverse of the forward path.
If forward and backward paths differ, \system{} will create flyover reservations in the backward direction only for the ASes in the overlap between the two paths.
We discuss the applicability of these requirements on paths characteristics to the Internet in \cref{sec_appendix_deployment_flyovers}.

\paragraph{Direct AS Interconnection}
In our model, adjacent ASes connect through direct links.
In case connectivity is provided through lower-layer interconnects at an IXP, we expect the IXP to deploy adequate measures against congestion in its switching fabric.
Alternatively, an IXP may explicitly participate in routing and offer Helia flyovers by registering as an AS.

\paragraph{Time Synchronization}
Lastly, we also require time synchronization on the order of \SI{100}{\ms} between ASes and their components~\cite{Arceo2009,Annessi2017SecureTime,Annessi2017,ptp,rfc5905}, which is necessary for replay suppression and to mark the validity periods of flyover reservations.

\medskip
We note that all these assumptions are in line with the requirements of other secure bandwidth reservation systems (\cref{sec_related}).

\subsection{Reservation Setup} \label{sec_reservation_setup}
In the setup phase, the source AS~$S$ requests authenticators from some or all on-path ASes.

\paragraph{Setup Request}
To this end, the source AS sends a single request packet containing its identifier~$S$, plus a reservation flag~$R_i$ and a backward reservation flag~$B_i$ for each on-path AS for which it wants to setup a flyover reservation. These could be all on-path ASes, or a subset of them, depending on its requirements. $R_i$ indicates to an \system{}-enabled AS whether it is expected to return an authenticator to the source AS in the direction of the request, and the backward reservation flag specifies whether an authenticator for the backward direction is required.
The source adds a current timestamp tsReq and authenticates the request for on-path AS~$i$ using the key \drkey{i}{S}:
\begin{align*}
    \text{Auth}_i = \mac{\drkey{i}{S}}{\text{tsReq}, R_i, B_i}
\end{align*}
The source identifier $S$ is not included in the computation of the MAC, because it is already implicit in the derivation of \drkey{i}{S}.
The setup request hence contains the following fields:
\begin{align*}
    \textsc{SetupReq} &= S, \text{tsReq}, (R_i, B_i, \text{Auth}_i) \\
    \forall i & \in A \subseteq \{1, \dots, \ell\}
\end{align*}
Here, $A$ refers to the subset of on-path ASes from which the source AS wants to request a reservation.
We denote by $A_F \subseteq A$ the set of on-path ASes for which the source AS requested a reservation in the forward direction, and by $A_B \subseteq A$ the set of on-path ASes for which it requested a reservation in the backward direction.
Instead of requesting flyovers from all on-path ASes in one packet, it is also possible for the source AS to send each on-path AS a separate request packet.
In the following, we will distinguish between ingress border routers, which receive a reservation setup packet from outside their AS, and egress border routers, which forward the packet from their AS to a neighboring AS.
Ingress border routers are responsible for the admission of flyovers for the forward direction, while egress border routers admit backward reservations.\footnote{The admission could also be performed in a dedicated reservation service inside the on-path AS. However, we consider admission on the border routers as a distinguished feature of \system{}, as this is not feasible in previously proposed secure reservation systems, which show significantly more computation overhead than \system{}.}

\paragraph{Bandwidth Admission}
One on-path AS after the other handles the request as specified in \cref{alg_admission}: upon receiving the request packet the ingress border router of on-path AS~$i$ performs the bandwidth admission (in case $R_i$ is set).
After checking that tsReq is current, i.e., within a small deviation from the system clock, the border router derives \drkey{i}{S} (\cref{eq_drkey}) and validates the authenticity of the request. If the validation succeeds, the border router computes the bandwidth guarantee \bw{i} it wants to provide to the source AS for the interface-pair (\emph{ing, egr}) that the packet is about to traverse.
Further, it specifies the time when the reservation should expire (\tsexp{}).
The router then derives the authenticator \auth{i}:
\begin{equation}
    \auth{i} = \mac{\key{i}}{S, ing, egr} \label{eq_auth}
\end{equation}
Note that the key for the MAC computation is the secret key \key{i}, only known to AS~$i$.
The border router then performs authenticated encryption with associated data (AEAD) with \auth{i} as the plaintext and (\bw{i}, \tsexp{i}) as the header (which is only authenticated and not encrypted), with the key \drkey{i}{S}.
It then adds the resulting ciphertext and tag, plus \bw{i} and \tsexp{i} to the packet. 
Surprisingly, all these operations can happen in the \dataplane{}---a setup packet is handled by the border router together with all other data traffic.
As a last step, the ingress
border router registers $S$, \bw{i}, and \tsexp{i} at the deterministic traffic monitor.
Similarly, if $B_i$ is set, the egress border router calculates and encrypts \auth{i,B}, where
\begin{equation}
    \auth{i,B} = \mac{\key{i}}{S, egr, ing}, \label{eq_auth_back}
\end{equation}
authenticates it together with \bw{i,B} and \tsexp{i,B}, and adds it to the packet.
If during the admission some verification check fails, the request is not dropped, but forwarded anyways.
This way, ASes later on the path will still receive the reservation request, and can independently perform the bandwidth admission.

\paragraph{Reservation Response}
The corresponding setup response
contains the encrypted authenticators $\auth{i}^{Enc}$, the granted bandwidth $\bw{i}$, and an expiration timestamp \tsexp{i}, as well as Tag$_{i}$ protecting the integrity of those three entries:
\footnote{The reservation responses can also be implemented such that one response packet is sent back per on-path AS, instead of one aggregated response.}
\begin{align*}
    \textsc{SetupResp} &= (\auth{i}^{Enc}, \text{Tag}_{i}, \bw{i}, \tsexp{i}),\\
    &~~~~(\auth{j,B}^{Enc}, \text{Tag}_{j,B}, \bw{j,B}, \tsexp{j,B}) \\
    &~~\forall i \in A_F, \forall j \in A_B
\end{align*}
The entries marked with an additional index $B$ correspond to the backward reservations. Depending on $R_i$ and $B_i$, the first, the second, or both tuples may be present in the setup response.

\paragraph{Reservation Storage}
After the request is processed by the last desired on-path AS, the request is simply sent back to the source AS without further processing by any AS on the return path.
Finally, the reservation service of the source AS verifies the authenticity of the authenticators, the reservation bandwidths, and the expiration times by validating the corresponding response tags. After decrypting the authenticators, it stores all this information for the later use in the transmission phase.

\begin{algorithm}[t]
\DontPrintSemicolon
  \KwInput{\key{i}, $S$, tsReq, $R_i$, $B_i$, Auth$_i$}
  \If{\upshape $!R_i$ or (now() - tsReq) $\notin [-\delta, L+\delta]$}
  {
  Forward packet
  }
  Calc. \drkey{i}{S} \tcp*{\cref{eq_drkey}}
  Calc. Auth$_i$ \\
  \If{\upshape Auth$_i$ is not correct}
  {
  Forward packet
  }
  Pass ($S$, TsReq) to duplicate-suppression system \\
  ($\bw{i}$, \tsexp{i}) $\leftarrow$ getBandwidth() \\
  \If{\upshape ($\bw{i}$, \tsexp{i}) = ($\bot, \bot$)}
  {
  Forward packet
  }
  Calc. authenticator \auth{i} \tcp*{\cref{eq_auth}}
  $\auth{i}^{Enc}, \text{Tag}_i \leftarrow \authenc{\drkey{i}{S}}{\auth{i}, [\bw{i}, \tsexp{i}]}$ \\
  Add $\auth{i}^{Enc}, \text{Tag}_i, \bw{i}, \tsexp{i}$ to the packet \\
  Register ($S, \bw{i}, \tsexp{i}$) at the traffic monitor\\
  Forward packet
\caption{Bandwidth Admission\\(at ingress border router of AS $i$)}
\label{alg_admission}
\end{algorithm}

\subsection{Data Transmission} \label{sec_data traffic}

\paragraph{Reservation Traffic Generation}
Once the source AS obtains the authenticators, its reservation service can use them to send data traffic over the bandwidth reservation (\cref{alg_send}).
To authenticate the traffic origin (itself), the reservation, and the length of the packet to the on-path ASes, the source calculates per-packet MACs:
\begin{equation}
    \varphi_{i} = \mac{\auth{i}}{\tspkt{},~\lenpkt{}} \label{eq_rvf}
\end{equation}
Here, \tspkt{} is a high-precision timestamp to uniquely identify each packet in order to prevent replay-attacks.
In the reservation packet, the source includes a truncated version of $\varphi_{i}$ called reservation validation field (RVF): RVF$_i = \varphi_{i}$[0:\lenvf{}].\footnote{\label{note_lenvf}The notation X[i:j] denotes the substring from byte i (incl.) to byte j (excl.) of X. We assume that the MAC function also provides the properties of a PRF, and suggest a value of \num{3} for \lenvf{}.
The MAC length is chosen to balance security and communication overhead. To brute-force the key or forge a MAC, an adversary must resort to an online brute-force attack, which is easily detectable. Further, a successful attack only allows the adversary to send a single packet over a single hop.}

If the source also wants to support backwards reservations, it additionally computes $\varphi_{i}^{Back}$:
\begin{equation}
    \varphi_{i,B} = \mac{\auth{i,B}}{\tspkt{},~\lenb{}} \label{eq_bvf}
\end{equation}
Similarly to the RVF, only the first \lenvf{} bytes of $\varphi_{i,B}$ are included in the reservation packet, which we denote as backward reservation validation field (BVF): BVF$_i = \varphi_{i,B}$[0:\lenvf{}].\cref{note_lenvf}
The \lenb{} field specifies the maximal length in bytes of such a response packet sent over the backward reservation. Thus, the source AS controls how much backwards traffic a destination is allowed to send. 
This is crucial because without this mechanism to limit return traffic the destination could abuse the backward reservation to cause bandwidth overuse, that would then be attributed \emph{to the source}.
To support a higher backwards packet rate compared to the forwards rate, Helia can add multiple timestamps and BVFs to the forward packet.

Helia therefore provides two complementary mechanisms to protect two-way communication: (i) the use of backward reservations, or (ii) the setup of a separate Helia (forward) reservation from the destination to the source. These mechanisms fit two slightly different use cases. 
Backwards reservations are most useful when the source expects a short reply, and may therefore not want to incur in the delay introduced by the destination having to create a reservation in the opposite direction. 
For high-bandwidth return traffic, however, we expect the destination to set up a dedicated reservation. Thus, the destination can better control the sending rate, and the source can avoid the overhead of sending multiple timestamps and BVFs in each forward packet.

A \system{} reservation packet has the following fields:
\begin{align*}
    \textsc{PktData} &= S, \text{D}, \tspkt{}, \lenb{}, \text{RVF}_i, \text{BVF}_j, \text{P} \\
    &~~~~~~~~~~~~ \forall i \in A_F, \forall j \in A_B
\end{align*}
The D-flag indicates whether the packet is sent in the forward (D=0) or in the backwards direction (D=1).

\paragraph{Reservation Traffic Forwarding}
In the forward direction, the ingress border router of AS~$i$ checks the validity of the RVF as specified in \cref{alg_validation}.
First, the router verifies whether the packet timestamp is current. If it is not, the packet is forwarded as best-effort, and therefore without any delivery guarantees.
Otherwise, the router recomputes \auth{i} based on \cref{eq_auth}, which it further uses to derive the RVF according to \cref{eq_rvf}.
It compares the computed RVF to the one included in the packet.
Again, if they do not match, the packet is forwarded with best-effort.
Finally, a duplicate-suppression system drops the packet in case it was replayed, and the deterministic traffic monitor checks that the packet does not cause an overuse of the flyover reservation.
In the backwards direction the process is very similar, with the difference that the D-flag is now enabled, meaning that it is the egress border router (from the view of the original source AS) that verifies the BVF.
The verification of the BVF is identical to the verification of the RVF, but instead of recomputing \auth{i} and $\varphi_i$, the router derives $\auth{i,B}$ and $\varphi_{i,B}$.
If for some ASes the RVF or BVF is not present or its validation fails, then the border router forwards the packet as best-effort.
The prioritization of successfully validated reservation traffic can be implemented in practice through queuing disciplines such as priority queuing. With such a mechanism, no bandwidth is wasted, as unused reservation capacity is dynamically reallocated to best-effort traffic.

\begin{algorithm}[t]
\DontPrintSemicolon
  \KwInput{$S$, $\auth{i}~(\forall i \in A \subset \{1, \dots, \ell\})$, $\auth{i.B}~(\forall j \in B \subset \{1, \dots, \ell\})$, lenB, P}
  pkt $\leftarrow$ (create new packet with reserved space for header and payload) \\
  pkt-len $\leftarrow$ len(pkt) \\
  tsPkt $\leftarrow$ time() \\
  \For{\upshape $i \in A_F$}
  {
  Calc. $\varphi_i$ \tcp*{\cref{eq_rvf}}
  }
  \For{\upshape $j \in A_B$}
  {
  Calc. $\varphi_{j,B}$ \tcp*{\cref{eq_bvf}}
  }
  Create packet with $S$, $D\leftarrow0$, \tspkt{}, lenB, RVF$_i$ ($\forall i \in A_F$), BVF$_j$ ($\forall j \in A_B$), P \\
  Send packet to first on-path AS

\caption{Sending Reservation Traffic\\(at reservation service of source AS)}
\label{alg_send}
\end{algorithm}

\begin{algorithm}[t]
\DontPrintSemicolon
  \KwInput{\key{i}, pkt}
  \If{\upshape now() - tsPkt $\notin [-\delta, L+\delta]$}
  {
  Forward packet as best-effort
  }
  Calc. \auth{i} \tcp*{\cref{eq_auth}}
  Calc. \rvf{i} \tcp*{\cref{eq_rvf}}
  \If{\upshape \rvf{i}[0:\lenvf{}] != RVF$_i$}
    {
    Forward packet as best-effort
    }
  Pass ($S$, \tspkt{}) to duplicate-suppression system \\
  Pass ($S$, len(pkt), ing, egr) to traffic monitor \\
  Forward packet (high priority)

\caption{Validating Reservation Traffic\\(at ingress border router of AS $i$, in the forward direction, i.e., for $D=0$)}
\label{alg_validation}
\end{algorithm}

\paragraph{Reservation Renewal}
As soon as a flyover reservation is established, the traffic using it is protected until its expiry.
It is then beneficial to renew the reservation, i.e., to send a setup request before the reservation expires, using the (bidirectional) reservation instead of best-effort, which guarantees a successful renewal.
Because an authenticator is neither based on the provided bandwidth nor the expiration time~(see \cref{eq_auth}), the authenticator does not need to be updated when the reservation is renewed.

\subsection{Deployment Considerations}

\paragraph{Allocation Matrix Changes}
So far we assumed a static topology with static allocation matrices. However, link failures, capacity upgrades, and the introduction of new links can alter the network topology. An AS must be able to respond to these events, without breaking the forwarding guarantees during that update.
In case an AS wants to decrease an entry of its allocation matrix, it directly uses the updated value in the computation of future bandwidth reservation sizes. To prevent congestion, the capacity between the two corresponding interfaces is only reduced after at least one reservation validity period. Increasing an entry in the traffic matrix is also possible, and in this case, an AS uses the updated value directly after increasing the network capacity for the corresponding interface-pair. Changes to one allocation matrix entry do not affect the size of reservations passing through other interface-pairs.

\paragraph{Monitoring}
Like other secure bandwidth reservation systems, \system{} requires a duplicate-suppression system plus a bandwidth monitor deployed at every on-path AS to filter out replayed packets and to detect reservation overuse by malicious source ASes.
Because previous systems are path-based, they are forced to rely on a probabilistic bandwidth monitor to police traffic: given the (potentially exponentially) many reservation paths that cross an interface pair, deterministic monitoring would require excessive amounts of memory.
In contrast, flyover providers in \system{} only have to police source ASes, irrespective of which paths the reservations take. 
\system{} can thus use a deterministic monitor based on token-buckets \cite{networkrouting2007} to precisely police the traffic of every AS.
As a token-bucket can be implemented using a single 8-byte timestamp (see \cref{sec_appendix_token_bucket}), even for \num{100000} ASes the required memory adds up to only \SI{800}{kB}.
A malicious AS overusing its flyover reservation can therefore be detected with certainty.

\paragraph{Infrastructure Fault Tolerance}
In \system{}, the setup and renewal of flyover reservations is done directly at the border routers, no communication with the reservation services of the on-path ASes is necessary.
This is in contrast to existing reservation systems, where a setup request not only traverses all border routers, but also the reservation services of the on-path ASes.
Having fewer infrastructure dependencies for the reservation setup reduces the setup time,  the computation and communication overhead, as well as overall protocol complexity.
In case of a border router failure, all communication that needs to pass the affected router, including reservation traffic, gets naturally dropped.
The bandwidth admission and the reservation data traffic forwarding only involve a few simple checks and computations, reducing the potential for implementation errors and failures. 
In case of a failure of the reservation service at the source AS, including the loss of all authenticators, flyovers for a path can be re-established in a single round trip.

\paragraph{Incremental Deployment}
Whereas in path-based reservation protocols every AS on a reservation path needs to support the reservation protocol, \system{} can be deployed incrementally: the reservation setup requests and data plane reservation fields are ignored by ASes not participating in the protocol. Such a selective processing of packets can be implemented through hop-by-hop extension header options, which are readily available in IPv6 and SCION~\cite{ipv6,scionhbh}.
These \emph{partial reservations}, which are only enforced by Helia-enabled ASes, are further discussed in \cref{sec_discussion_emerging_properties}. 
Partial reservations offer increasingly higher protection (compared to best-effort forwarding), until all on-path ASes deploy Helia and full end-to-end forwarding guarantees are achieved.
Finally, an AS only needs to update the set of border routers through which reservations should be supported.
Internal routers do not need to be modified.
\newcommand{\acover}[0]{$\gamma$-cover}
\newcommand{\bwdiv}[0]{concurrent flyover algorithm}
\newcommand{\bwdivlong}[0]{Concurrent flyover algorithm}
\newcommand{\tmdiv}[0]{maximum flyover algorithm}
\newcommand{\tmdivlong}[0]{Maximum flyover algorithm}

\section{Are flyovers large enough?} \label{sec_simulations}
We investigate this question by simulating the end-to-end flyover reservation
sizes on a range of network topologies.
In particular, we show that flyover reservations are large enough to carry
\emph{critical communication traffic}. For comparison, we also compute the size
of GLWP reservations.
(\cref{sec_background_reservation_systems}).

\subsection{Reservation Algorithms}
\paragraph{Flyovers}
Source ASes requesting reservations from remote flyover providers obtain bandwidth shares computed as in \cref{eq_bandwidth}.
They can then freely allocate the flyover bandwidth among protected paths (\cref{sec_traffic_control}).
This decision is rooted in local policies, as the AS could, e.g., decide to grant more of the reserved resources to higher-priority communication paths.
We choose to evaluate two \emph{source allocation strategies} at the extremes of the policy space:
\begin{itemize}
    \item \emph{\bwdivlong{}}: The source AS assigns an equal share of the flyover reservation for an interface pair to each protected path traversing it.
    \item \emph{\tmdivlong{}}: The source AS assigns the full bandwidth of the flyover to one of the protected paths at a time, maximizing the end-to-end reservation size for each path, but preventing the concurrent use of the flyover.
\end{itemize}
The size of the final reservation for a path is given by the minimum of the shares of all on-path flyover reservations.

\paragraph{GLWP}
We also compare against the reservation sizes computed by the algorithm used in GLWP, for a total of three different path reservation algorithms. We do not simulate Colibri's allocations as they depend on the concrete traffic demands in the network, as well as other complex topological features that are outside the scope of our simulation (such as SCION isolation domains~\cite{SCIONBookv2}).

\subsection{Simulation Setup}

\paragraph{Topologies}
We model networks as graphs, with nodes representing ASes in the Internet and unidirectional weighted edges representing inter-AS links. Each node is augmented with external and internal interfaces and an allocation matrix.
We simulate the output of the three reservation strategies on randomly-generated
Barabási–Albert graphs.\footnote{As implemented in the FNSS library~\cite{fnss}. Simulations are run using Snakemake~\cite{snakemake}.} We chose this
particular type for random graphs as it produces scale-free graphs, which are
found to closely approximate the structure of technological
networks~\cite{Broido2019}, the Internet in particular. 

\paragraph{Link Capacity Model}
We assign bandwidth to links in each topology based on a degree-gravity model, for which links connecting nodes with higher degrees receive a higher bandwidth. 
This is usually the case in real networks, where large, central nodes typically meet higher traffic demands with higher-capacity infrastructure.
Link bandwidths are selected from the range \num{40}--\SI{400}{Gbps} in increments of \SI{40}{Gbps}, which are representative of inter-domain link bandwidths.\footnote{These sizes are arguably conservative: Today's Internet undersea cables can carry hundreds of Tbps~\cite{marea}.}
The bandwidth of the internal link---i.e., the link connecting the external interfaces and the internal interface of a node---is set as the maximum capacity of all external links. 

\paragraph{Allocation Matrices}
The allocation matrix $\text{M}_i$ for each node~$i$ is computed starting from the capacities of the links $\text{C}_a$ attached at each interface $a$. 
First, entry $(a, b)$ in the matrix is initialized with the capacity $\text{C}_b$ of the edge attached to the node at interface~$b$. 
Then, each column of the matrix is normalized to the capacity of the outgoing edge, such that $\sum_{a} \text{M}^{a,b}_i = \text{C}_b$ for all interfaces $b$.
Finally, if any of the row sums exceed the capacity of an incoming edge, the relative entries are further normalized such that $\sum_{b} \text{M}^{a,b}_i \leq \text{C}_i$. 

\paragraph{Destination Sampling} 
The flyover bandwidth allocation algorithm (\cref{eq_bandwidth}) takes advantage of the fact that an AS is usually
communicating with only a fraction of the total number of ASes at any given
moment. We therefore run simulations in which each AS communicates only with a
fraction $r$ of ASes. Unfortunately, however, realistic estimates of $r$ are
not readily available. Therefore, we evaluate the reservation algorithms for different values of $r$.
For every value of $r$ and for each source node, we randomly sample
destinations skewing the distribution in favour of nodes with a higher degree:
For every value of $r$ and for each source node, we sample $d=\nint{ r \cdot N}
$ destinations, where the probability of sampling a destination node~$i$ is
$p_i = \text{deg}_i/\sum_j \text{deg}_j$. 

\paragraph{Metrics: Reservation Size \& Cover}
We are interested in showing which of the three algorithms consistently provides the largest reservations.
We first compare the distribution of reservation sizes in each graph.
However, this does not well represent the nature of critical traffic, for which either there is enough reserved bandwidth available, or the reservation is useless.

To capture this binary requirement on critical traffic, the original GLWP paper introduces the notion of \acover{}. Intuitively, the \acover{} of a node is the fraction of other nodes it can reach at any time with a reservation of size at least $\gamma$, which is called the \emph{cover threshold}. 
Formally, given all reservations computed by algorithm $A$ on a graph with nodes $V$, $|V|=N$, the cover for node~$i$ is:
\begin{align}
    \gamma\text{-cover}_i^A = \frac{\sum_{j \in V} I(a_{ij} > \gamma)}{N}
\end{align}
where $I(\cdot)$ is an indicator function, and $a_{ij}$ is the size of the reservation from node~$i$ to node~$j$.
However, flyovers only consider a fraction $r$ of destinations. 
Therefore, we have to adapt the previous definition to consider the set $S_i$ of sampled destination nodes ($|S_i|= \nint{r \cdot N}$) for each node $i$:
$\gamma\text{-cover}_i^A = \frac{\sum_{j \in S_i} I(a_{ij} > \gamma)}{\nint{r \cdot N}}$.
We relax the definition of \acover{} to accept that nodes in the cover may not be reachable simultaneously, which is the case for the \tmdiv{}.

\subsection{Results}

\begin{figure}[t!]
    \centering
    \begin{subfigure}{\linewidth}
        \centering
        \includegraphics[width=\linewidth]{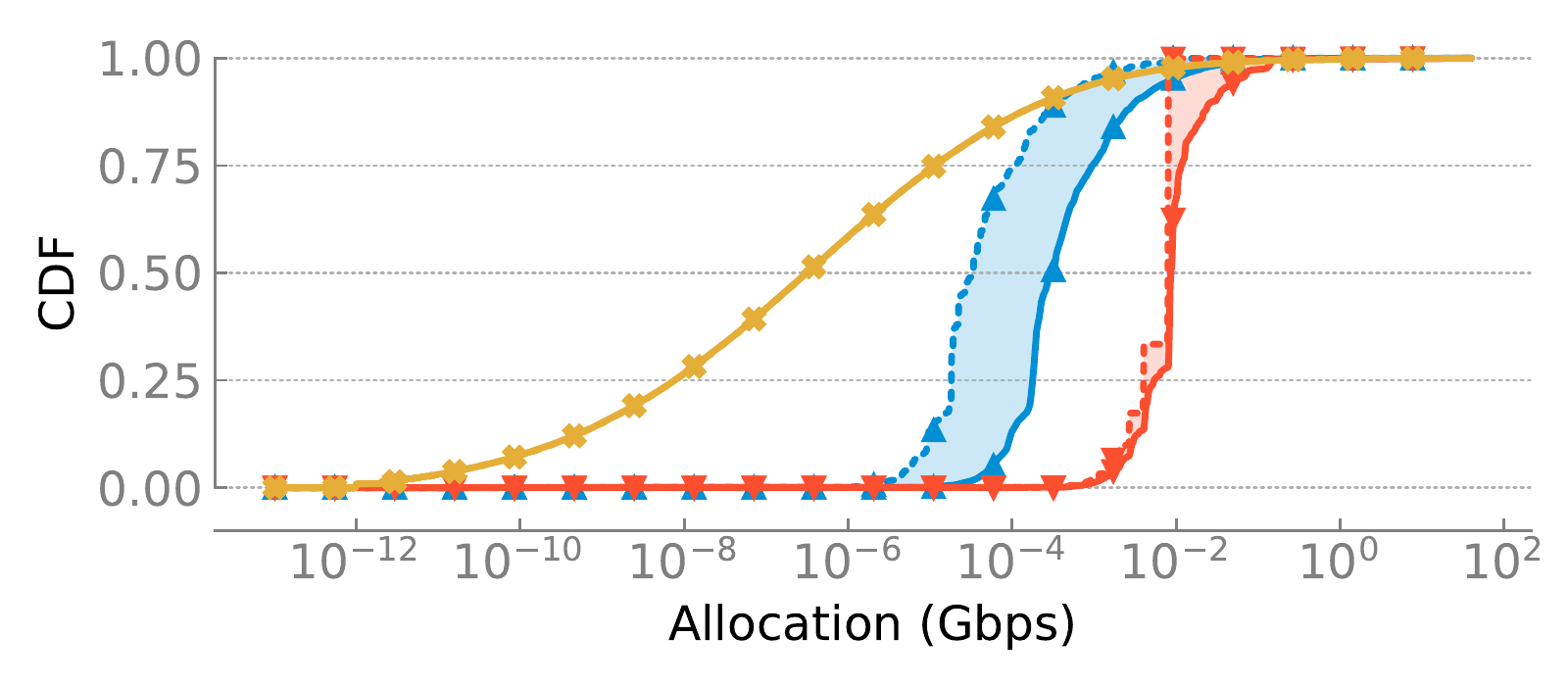}
        \caption{Reservation size distribution.}
        \label{fig:allocation_distribution} 
    \end{subfigure}

    \begin{subfigure}{\linewidth}
        \centering
        \includegraphics[width=\linewidth]{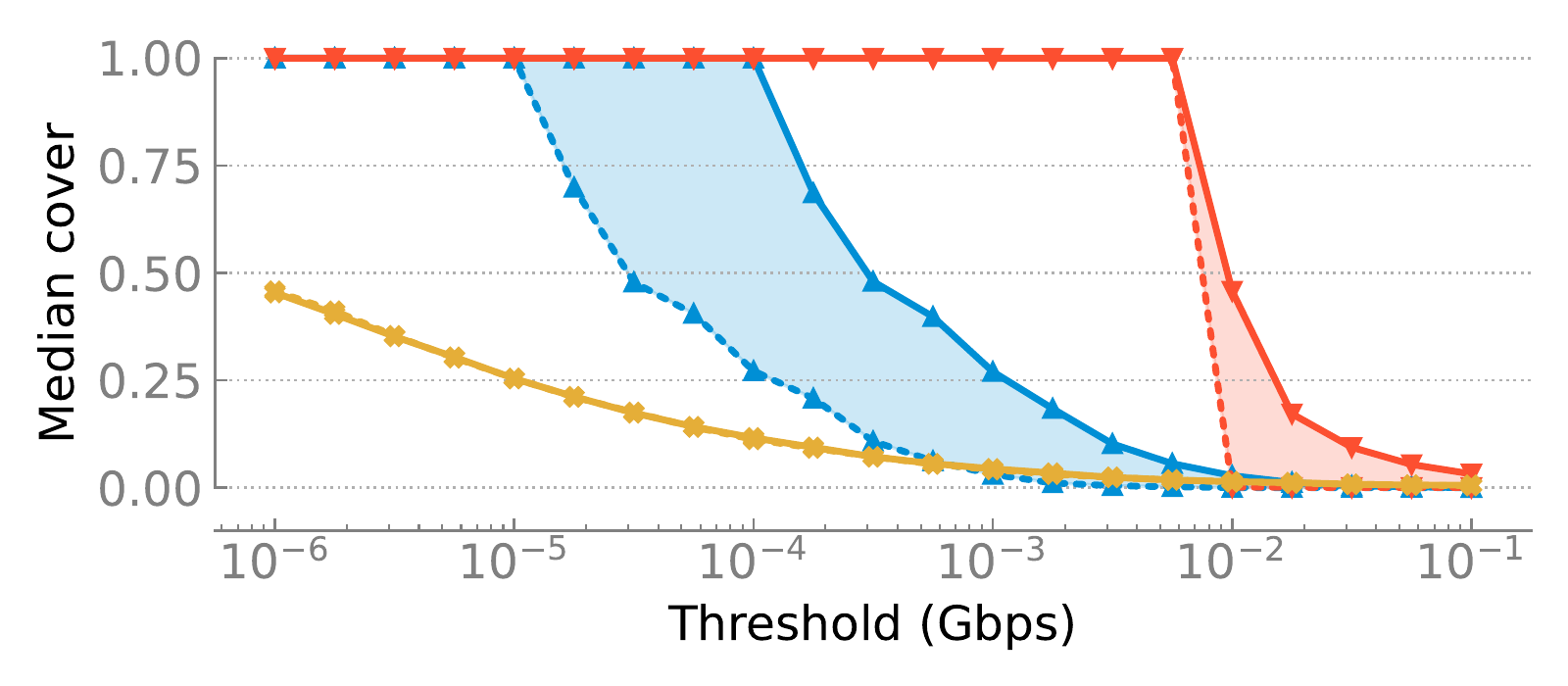}
        \caption{Median \acover{} achieved with increasing cover thresholds.} 
        \label{fig:cover_thrs}
    \end{subfigure}
    
     \begin{subfigure}{\linewidth}
        \centering
        \includegraphics[width=\linewidth]{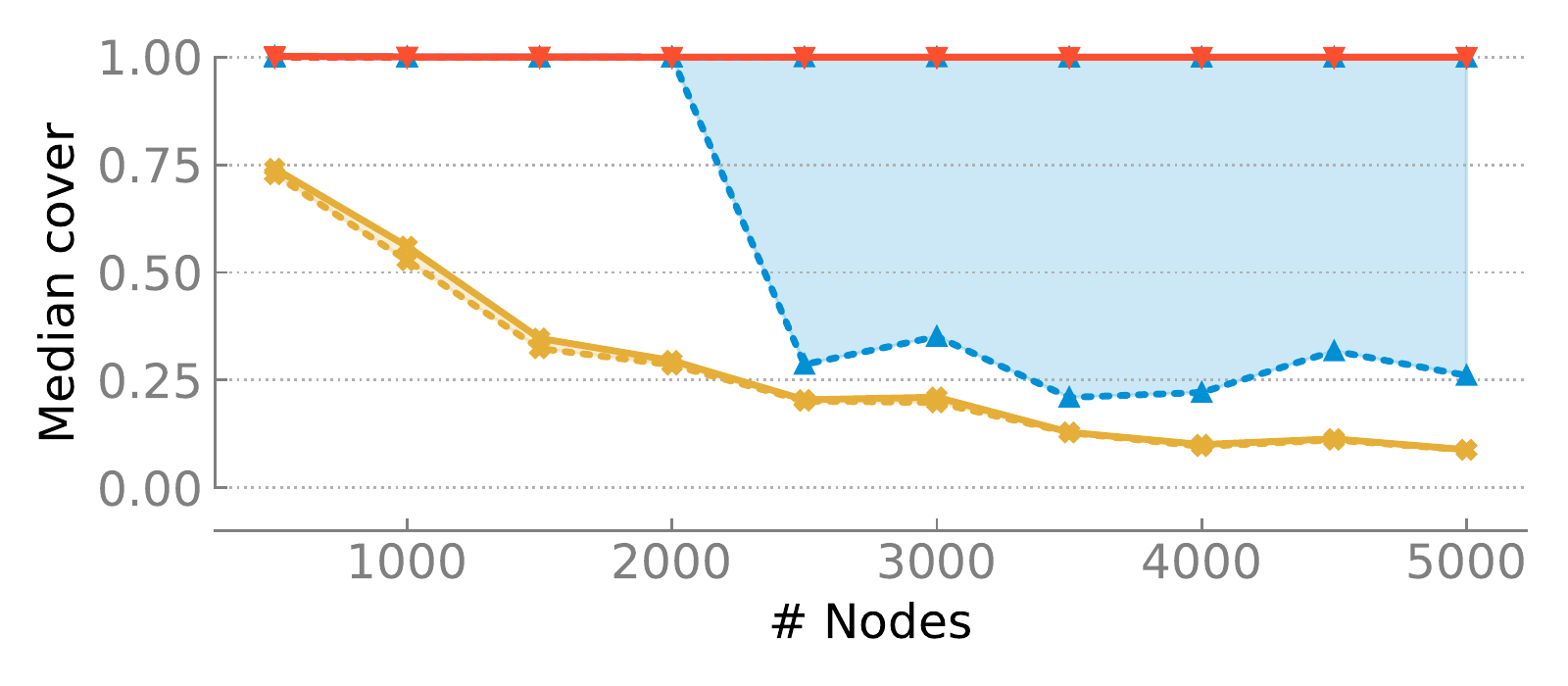}
        \caption{Median 100-kbps cover, varying the size of the graphs.}
        \label{fig:cover_nodes}
    \end{subfigure}   
    
    \begin{subfigure}{\linewidth}
        \centering
        \includegraphics[width=\linewidth]{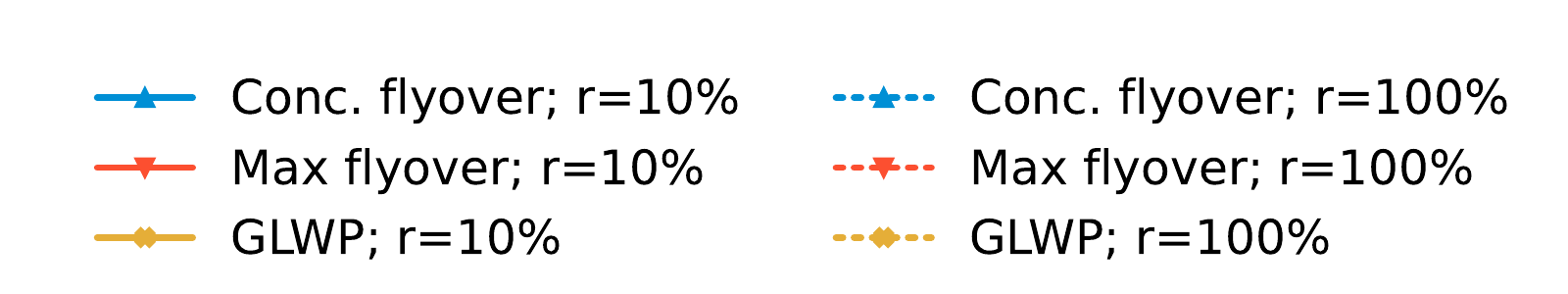}
    \end{subfigure}

    \caption{Simulation results for a 5000-node graph. The shaded areas represents the interval between the 10\% and 100\% destination sampling rate.
    }
    \label{fig:simulation_results}
\end{figure}

We report here the results of three experiments that are illustrative of the differences between the algorithms.

\paragraph{Reservation Size Distribution}
\Cref{fig:allocation_distribution} shows the cumulative distribution function (CDF) of the reservation sizes on a \num{5000}-nodes graph for the three algorithms and varying sampling rates.
Confirming intuition, we notice that lower sampling rates yield higher reservation sizes compared to the higher sampling rates.
We further observe that GLWP reservation sizes are lower, and distributed over a much wider range of values, compared to the reservation sizes produced by assembled flyover reservations.
The median reservation computed with the maximum flyover algorithm is more than \num{10000}$\times$ larger than the median GLWP reservation, which is due to GLWP's reservation bandwidths decreasing exponentially in the path length (\cref{sec_related_GLWP}).
This difference is also highlighted in the next experiment.

\paragraph{Median $\gamma$-Cover}
\Cref{fig:cover_thrs} shows the median \acover{} across all nodes for $\gamma$ varying 
from \SI{1}{kbps} to \SI{100}{Mbps}, computed on the same \num{5000}-node graph.
The \tmdiv{} provides high reservation sizes, and therefore high median cover, while the \bwdiv{} only achieves full cover up to a threshold of \SI{100}{kbps}, and degrades to zero around the \SI{10}{Mbps} threshold. Finally, GLWP's median cover is generally substantially lower, although it degrades less rapidly.

\paragraph{Cover Scalability}
The median \SI{100}{kbps}-cover of the three algorithms, computed on random graphs of increasing size from \num{500} to \num{5000} nodes, is shown in \cref{fig:cover_nodes}. We observe that the \tmdiv{} provides a \SI{100}{\percent} cover with this threshold even at the highest number of nodes. 
The \bwdiv{} has a broad range of median covers, depending on the selected sampling rate. As we also observe in \cref{fig:cover_thrs}, this algorithm provides full cover for $r=\SI{10}{\percent}$, and drops to little more than \SI{20}{\percent} for $r=\SI{100}{\percent}$. 
The sharp drop for $r=\SI{100}{\percent}$ at around \num{2000} nodes is due to the binary nature of the cover threshold: The increased number of communication paths between nodes increases contention, reducing the end-to-end flyover bandwidth for many paths below the threshold. 
This effect is typically more pronounced for longer paths. For shorter paths---connecting nodes at the edge and towards the core---the end-to-end flyover bandwidth remains instead within the cover threshold, as they do not have to compete with many other paths for bandwidth.
Further increasing the size of the graph does not significantly impact the reservation size for these short peripheral paths, leading to fairly static cover for graph sizes above \num{2000} nodes.

\subsection{Discussion} 
From our simulation results we can distill the following general statements.
First, an AS seeking to connect to multiple destinations concurrently using up
to \SI{100}{kbps} of reserved bandwidth can employ the \bwdiv{} version.
Second, for extremely critical communications, an AS can use the \tmdiv{} and
sacrifice simultaneous communication to reach another node in the network with a reservation of at least \SI{10}{Mbps}. Using the \tmdiv{} can thus greatly extend the
range of reserved-bandwidth communication.
In simulated settings, flyover-based algorithms greatly outperform GLWP.

\paragraph{Relevance to the Internet}
The super-quadratic computational complexity of the simulations here presented
limits the size of the graphs we could consider. The graphs used in the
simulations are thus small compared to the size of the Internet---roughly an
order of magnitude smaller.
Nevertheless, these graphs are more challenging for reservation algorithms than
the actual Internet.
The links used in the simulations are conservatively sized when compared to the
multi-Tbps links that compose today's Internet backbone. This additional
available capacity will proportionally increase the size of flyover
reservations. Then, the Internet core is densely interconnected, and pairs of
ASes are often connected by multiple links. This multi-link connectivity is not
reflected in our network models, and further increases the bandwidth available
for reservations. Finally, the number of ASes in the Internet increases
linearly~\cite{cidr} while network capacity grows
exponentially~\cite{routerStatistics}.

\begin{figure}[t!]
  \centering
\begin{subfigure}{\linewidth}
    \centering
  \includegraphics[width=\linewidth]{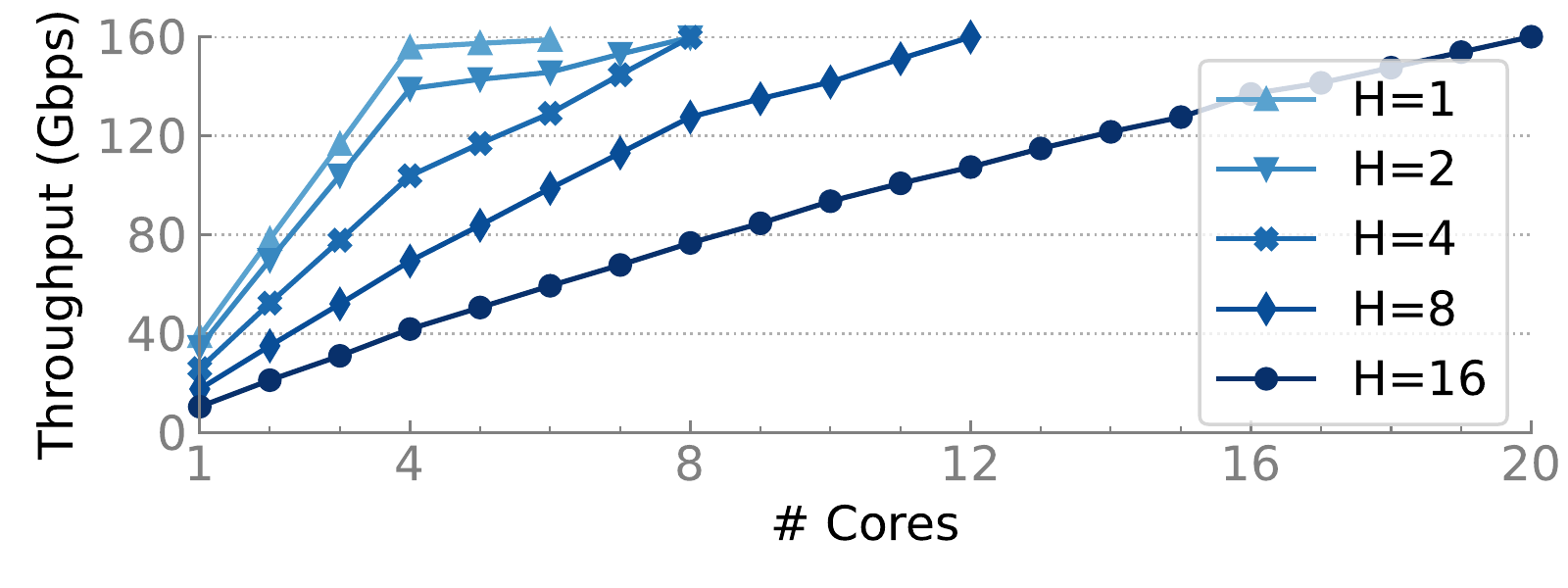}
  \caption{Unidirectional reservations.}
\end{subfigure}
\begin{subfigure}{\linewidth}
    \centering
  \includegraphics[width=\linewidth]{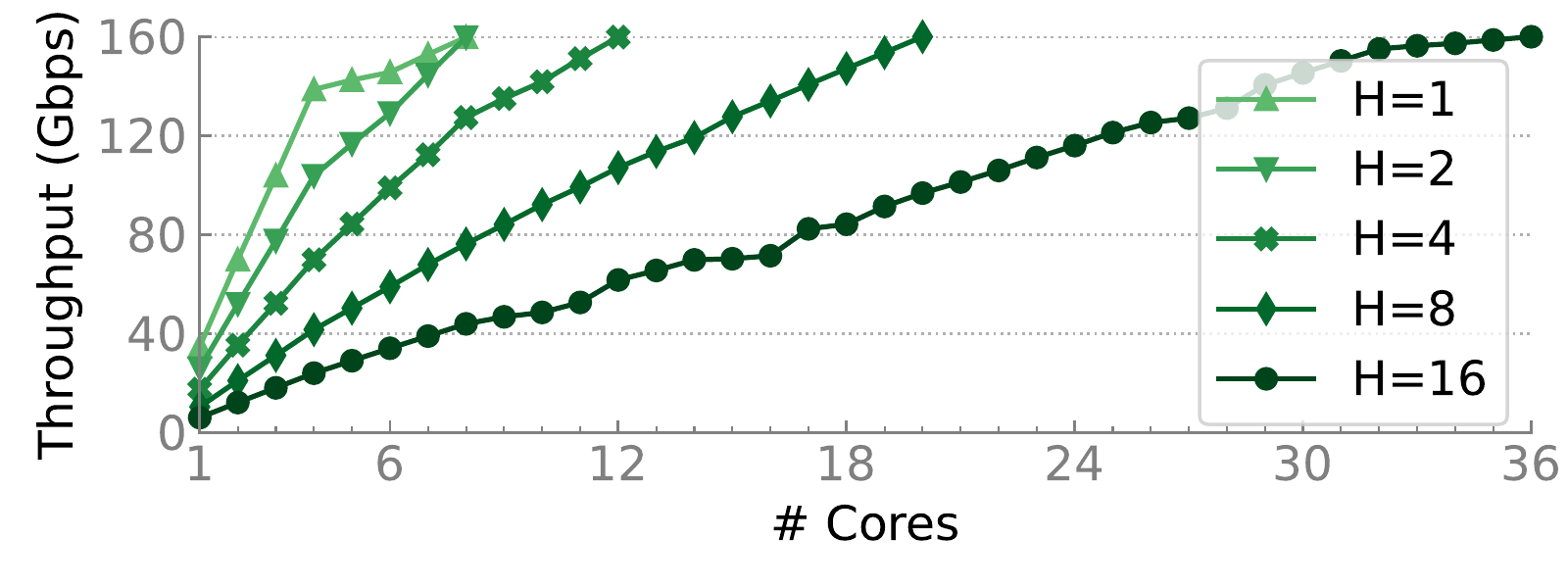}
  \caption{Bidirectional reservations.}
\end{subfigure}
  \caption{\textbf{Evaluation of \cref{alg_send}.} Packet generation performance for traffic in gigabits per second at the reservation service of the source AS for different numbers of cores and AS-level hops (H). The packets carry a payload of \SI{1000}{\byte}.
  }
  \label{fig_source_forward_gbps}
\end{figure}

\begin{figure}[t!]
  \centering
    \includegraphics[width=\linewidth]{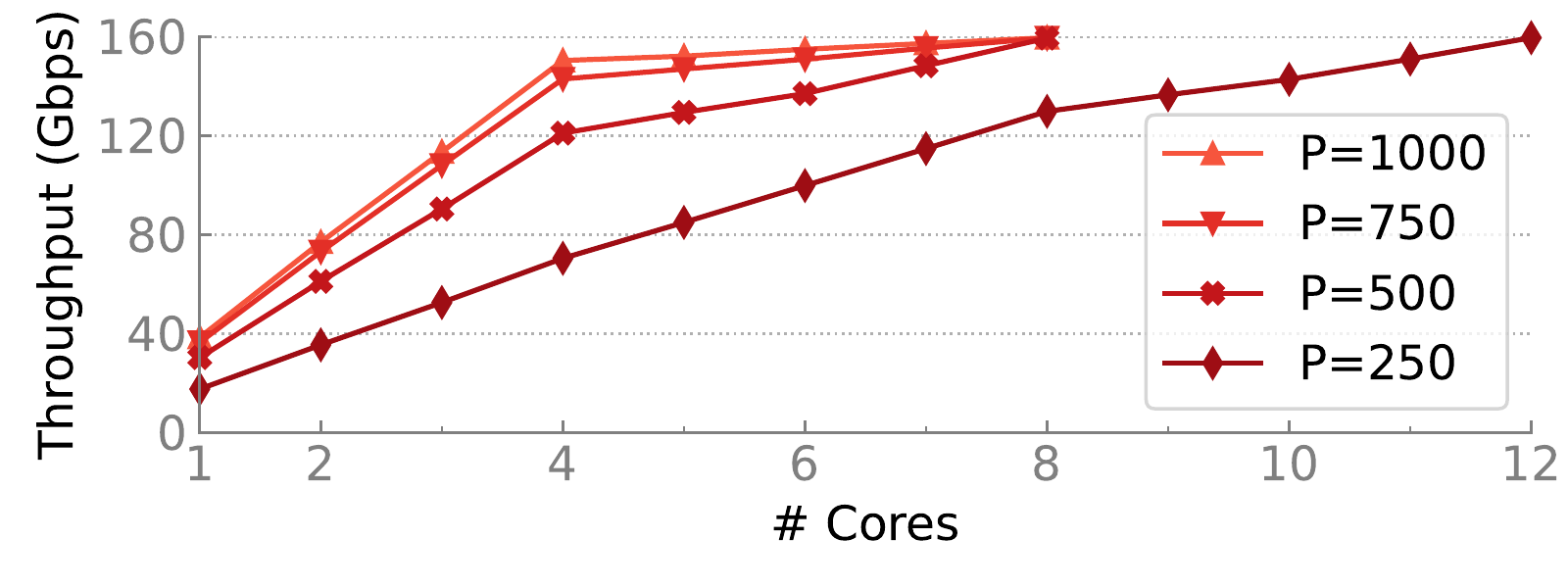}
  \caption{\textbf{Evaluation of \cref{alg_validation}.} Forwarding performance in gigabits per second at an ingress border router for different numbers of cores and payload sizes~(P).}
  \label{fig_router_forward_gbps}
\end{figure}

\section{Implementation and Evaluation}
\label{sec:implementation}
We implement and evaluate the bandwidth admission (\cref{alg_admission}), the reservation traffic generation (\cref{alg_send}), and the packet validation (\cref{alg_validation}) procedures, in order to demonstrate high performance at every component.
As forward and backward traffic validation procedures are identical, we only evaluate the packet validation performance at an ingress border router.
A performance comparison with Colibri and GLWP is provided in \cref{sec_appendix_comparison_performance}.

\subsection{Implementation}
We implement all three algorithms using Intel DPDK~\cite{dpdk}. \footnote{We chose DPDK because (i) it enables rapid prototyping of dataplane protocols; (ii) it is supported by a large range of devices~\cite{dpdk_supported}; and (iii) it is used by various manufacturers of both BGP/IP~\cite{dpdk_ip_border_router,dpdk_ip_border_router2} and SCION border routers~\cite{dpdk_scion_border_router}. Yet, any production-ready implementation of Helia does not necessarily need to rely on DPDK.}
In particular, we leverage the Bloom filter implementation in the DPDK membership library; 
the AEAD is based on the IETF variant of the ChaCha20-Poly1305 construction provided by libsodium~\cite{libsodium};
and we use the AES-128 block cipher in CBC mode for MAC functions. 
To speed up the AES computations, we rely on Intel’s AES-NI~\cite{AESni} hardware instructions available on all modern Intel CPUs.

The bandwidth admission uses the algorithm specified in \cref{sec_appendix_fine_grained} to compute $\rho$.
For deterministic monitoring and policing the reservations we use a hash table, mapping an AS number to its dedicated token bucket. Each token bucket is represented by a single \SI{8}{\byte} timestamp (\cref{sec_appendix_token_bucket}).
We do not explicitly include the path (a list of interface pairs) in the packet, as we assume that ingress border routers know the egress interface to which each \system{} packet must be forwarded, based on the destination address, local policies, and possibly other information in the header.

The duplicate suppression system is directly taken from the literature~\cite{replay2017}, and we do not include it in our evaluation.

\subsection{Testbed}
Our measurement setup consists of two machines, a commodity machine with an Intel Xeon 2.1 GHz CPU, which runs our algorithm implementation that is to be evaluated, and a Spirent SPT-N4U, which serves both as traffic generator (when evaluating the bandwidth admission and the traffic validation) and bandwidth monitor (when evaluating the traffic generation). Both machines are connected by four \SI{40}{Gbps} bidirectional Ethernet links. We evaluate all three implementations separately one after the other, they never run at the same time on the commodity machine.

\subsection{Results}

\paragraph{Bandwidth Admission}
The total duration of the bandwidth admission is \SI{618}{ns} on average.
This measurement also includes the overhead induced by the algorithm to compute $\rho$. In particular, we used \cref{alg_fine_grained} from \cref{sec_appendix_fine_grained} to implement the function \texttt{getBandwidth()}. The overhead caused by this computation amounts to \SI{30}{ns}, four orders of magnitude lower than the admission computation overhead of some of the previously proposed systems (\cref{sec_appendix_comparison_performance}).
In comparison, the most expensive operation is the authenticated encryption with \SI{411}{ns}.
The overall reservation admission overhead in \system{} is similar to the per-packet overhead of other data-plane protocols~\cite{EPIC,PPV}.
The admission process can be parallelized by running it on multiple cores, enabling a border router to handle reservation requests at line rate.

\paragraph{Traffic Generation}
The traffic generation performance of the source AS for unidirectional and bidirectional traffic is shown in \cref{fig_source_forward_gbps}. For packets with a payload of \SI{1000}{\byte} and for a path with eight per-hop reservations, one core of the source AS achieves a throughput of \SI{17.62}{Gbps} for unidirectional and \SI{10.40}{Gbps} for bidirectional reservation traffic.
This is already enough for most ASes and the critical applications they are running, especially when considering that the average path length in the Internet is around five AS-level hops~\cite{Wang2016, Bottger2019}. Still, the throughput can be linearly increased by simply dedicating more cores to the reservation traffic generation.
Because our implementation assigns only a single one of the four ports to each core, the throughput gain is lower for the last four cores added before achieving \SI{160}{Gbps}.

\paragraph{Traffic Validation}
\Cref{fig_router_forward_gbps} shows the throughput achieved by a border router validating and forwarding incoming reservation traffic.
For packets with \SI{1000}{\byte} payload, a single core can process \SI{38.48}{Gbps} of reservation traffic, which is close to the \SI{40}{Gbps} capacity of the Ethernet links connecting our measurement machines (a core is only assigned to a single link port).
With only \num{12} cores the router is able to saturate all links for a total of \SI{160}{Gbps} even in the case of short \SI{250}{\byte} payload packets.
Expressed in packets per second, one core can process \SI{6.45}{Mpps}, which is at least a factor of two improvement over the state of the art~\cite{Colibri,GLWP}.
A fine-grained analysis of the router processing overhead can be found in \cref{table_router_tasks}. All cryptographic operations take between \num{22} and \SI{38}{ns}, the total processing time is \SI{121}{ns}.

\paragraph{Deterministic Monitoring and Policing}
For packets with a payload of \SI{1000}{\byte}, one core of the deterministic bandwidth monitor is able to check and police reservations at a rate of \SI{37.81}{Gbps} even for traffic originating from one million different ASes.
This result can be further improved by using multiple cores that each handle a subset of all source ASes (every core has its own hash table), letting the system scale linearly in the number of cores in use.
The memory overhead is 8 bytes per monitored AS. LOFT~\cite{loft}, the state-of-the-art probabilistic monitor used by Colibri and GLWP, shows similar processing speeds but must keep per-flow state in the slow path.

\begin{table}[t!]
    \centering
    \caption{
        \textbf{Evaluation of \cref{alg_validation}.} Breakdown of the average packet processing time in nanoseconds at the ingress border router of an on-path AS.
    The processing time is independent of the reservation path length and payload size.
    }
    \begin{tabular*}{\linewidth}{l@{\extracolsep{\fill}} c} 
         \toprule
         Task & Avg. proc. time (ns)\\ 
         \midrule
         Verify \system{} packet header format & 24\\ 
         Validate timestamp is current& 7\\
         Calculate $\alpha$ using \cref{eq_auth} & 30\\
         AES key expansion of $\alpha$ & 38\\
         Verify RVF using \cref{eq_rvf} & 22\\ 
         \textbf{Total} & \textbf{121} \\
         \bottomrule
    \end{tabular*}

    \label{table_router_tasks}
\end{table}

\section{Security Analysis}

\subsection{Threat Model}

\paragraph{Security Objectives}
\newcommand{\reqirement}[1]{\textbf{R#1}}
With this analysis, we aim to show that Helia is a viable solution to provide strong traffic delivery guarantees, irrespective of attackers. This goal translates into the following low-level security requirements. An adversary must not be able to:
(\reqirement{1}) trick the system into increasing its allocated bandwidth share, nor into reducing the share of other ASes; (\reqirement{2}) indefinitely prevent ASes from setting up a flyover reservation; (\reqirement{3}) utilize the reservation granted to a different source AS; (\reqirement{4}) drop or significantly delay legitimate reservation traffic; and (\reqirement{5}) use reservations outside their validity period or beyond their maximum bandwidth.

\paragraph{Adversary Model}
An on-path adversary can trivially break our security goals by modifying, delaying, or dropping packets.
Therefore, we restrict the adversary model accordingly.

Helia's security properties hold only for paths consisting
of honest and uncompromised ASes. The adversary can nevertheless observe, inject, and replay packets on the links connecting on-path ASes.
Further, we do not place any restrictions on the capabilities and number of off-path adversaries. An off-path adversary can observe, modify, inject, drop, delay and replay packets, and even compromise entire ASes, including their control services, border- and internal routers, and get access to AS key material.
We finally assume that cryptographic primitives such as MACs and PRFs are secure,
and that the probability of false negatives in the duplicate-suppression system is negligible---as it can be driven arbitrarily low with additional resources.

\subsection{Attacks against Reservation Admission} \label{sec_attacks_admission}
\system{} defends against bogus admission requests by authenticating every packet, and against replayed requests thanks to the duplicate-suppression system. The authenticators are encrypted to protect against potential observers, as source authentication depends on their confidentiality.
In the following, we discuss \system{}'s resistance against further attacks targeting reservation admission.

\paragraph{Repeated Reservation Requests (\reqirement{1})}
An adversary may try to issue multiple authentic reservation requests to obtain multiple reservations for the same flyover, and thus more bandwidth than intended.
However, the authenticators and reservation sizes the source AS receives are the same, irrespective of the number of requests the source sends (\cref{eq_auth}).
Therefore, the bandwidth monitor of the provider AS only attributes reservation data traffic to one single reservation, preventing this attack.

\paragraph{Sybil Attacks (\reqirement{1})}
By registering a large number of ASes, an adversary can try to get a higher aggregate of \system{} bandwidth, while at the same time reducing the reservation size for honest ASes.
However, such an attack is cumbersome and likely to be detected
because AS numbers must be requested from regional Internet registries through a manual procedure~\cite{rfc1930,ripe-asn-request}
Furthermore, an adversary would need to set up the necessary components and keys to support BGP and RPKI.
If \system{} is used in a future Internet architecture such as SCION,
the computation of the reservation can be adapted to a per-ISD/per-AS division of the bandwidth, i.e., an allocation matrix entry is first divided by the number of ISDs, and only then on a per-AS basis.

Another mitigation against sybil attacks is their cost, as reservations are expected to be granted only in return for payment (\cref{sec_billing}).
Finally, ASes can define ``allow lists'' for their flyovers by defining access policies on which source ASes are allowed to obtain the symmetric keys necessary to issue a reservation request.

\paragraph{Bandwidth Admission Request Flooding (\reqirement{2})}
Computationally expensive services accessible by an unpredictable number of potentially untrusted sources need to be protected against flooding attacks, where one or multiple adversaries try to exhaust the computational resources of the service, with the objective of rendering the service unavailable for honest users.
Such flooding attacks do not have any impact on legitimate requests, regardless of whether the adversaries' requests are authentic or not.
In fact, as shown in \cref{sec:implementation}, validating and admitting a reservation request in \system{} is highly efficient, and packets can be processed at line rate. Therefore, no honest request is ever discarded.

\paragraph{DoC Attacks (\reqirement{2})}
In a denial-of-capability (DoC) attack the network is flooded with excessive amounts of traffic in order to prevent one or multiple source ASes from obtaining an authenticator.
As flyover reservations can be renewed over existing reservation traffic, and because reservation traffic can not be disturbed by any DDoS attempts, only the very first reservation request sent over best-effort traffic can actually be attacked.
Moreover, an adversary can try to prevent the source AS from fetching symmetric keys during the execution of the DRKey protocol. Because those keys can be requested already in advance, an adversary needs to successfully attack the exchange channel for a considerably long period of time.
Finally, the Docile~\cite{docile} system may be used to protect both the key exchanges and the initial reservation setup request without introducing any other avenues for DoC attacks.

\subsection{Attacks against Reservation Traffic}

\paragraph{Framing Attacks (\reqirement{3})}
With spoofed source AS identifiers in the reservation packets, an adversary can try to frame a honest source AS by causing bandwidth overuse at provider ASes, which could possibly result in sanctions against the source AS.
Another implication of a successful framing attack is the rejection of a portion of the reservation traffic at provider ASes that already incorporated the spoofed packets in the bandwidth monitoring.
However, \system{} prevents spoofing of source AS identifiers in the reservation data packets by means of per-packet source authentication.
An authenticator \auth{} used in the computation of an RVF or BVF is only known by one specific and legitimate source AS, where encryption ensures confidentiality of the authenticators in the reservation setup.

\paragraph{DDoS against QoS (\reqirement{4})}
As soon as an end-to-end reservation is established, packets sent over this channel profit from forwarding guarantees and enhanced quality of service (QoS).
To perturb traffic in this reservation channel, an off-path adversary can attempt to cause congestion at border routers by sending large volumes of attack traffic, which could lead to legitimate reservation packets being delayed or dropped.
However, volumetric attacks based on best-effort traffic have no impact, as border routers forward reservation traffic with strictly higher priority.
Similarly, attacks based on legitimate reservation traffic can also not impact other reservations in the same AS:
All reservations are policed so that they can use at most the portion of bandwidth allocated for them, while definitions of the flyover computation and the allocation matrix ensure that the sum of all reservation sizes never results in an over-allocation of any intra- or inter-AS link.

\paragraph{Reservation Overuse (\reqirement{5})}
Bandwidth overuse by a malicious or misconfigured source AS is detected and policed by deterministic bandwidth monitoring systems deployed at every on-path AS.

Tricking the monitoring system by reusing the same packet header, in particular with the same timestamp and the same RVFs, for multiple data packets with different payloads, is not possible.
The bandwidth monitor receives packets only after they have been vetted by the duplicate-suppression system, which discards packets with a combination of source AS identifier and timestamp it has already seen.
It is up to the AS observing overuse attempts to decide whether to notify or even sanction the misbehaving source AS, for example by denying future reservations.

\paragraph{Sending beyond Expiration (\reqirement{5})}
Flyover authenticators are valid until the underlying key established by the DRKey system is rolled over (\cref{eq_auth,eq_auth_back}).
Nevertheless, a malicious or misconfigured source AS still cannot send valid reservation traffic during periods for which no reservation was requested, e.g., after a reservation expired and was not actively renewed.
The reason is that, while the source AS can still generate packets with valid RVFs and BFVs, the deterministic bandwidth monitor is aware of the exact expiration times. Packets using expired reservations are then forwarded as best-effort traffic, i.e., without any delivery guarantees.

\section{Discussion}

\subsection{Emerging Properties} \label{sec_discussion_emerging_properties}
We highlight two novel use-cases of \system{}'s reservations, enabled by the simplicity and composability of flyovers.

\paragraph{Partial Reservations}
Instead of requesting multiple flyover reservations using a single setup packet, an AS can alternatively send a setup packet for every on-path AS independently.
Taking this observation one step further, the AS can also use flyover reservations to only protect selected parts of the whole forwarding path, instead of assembling flyover reservations to a full end-to-end reservation.
On-path ASes for which the source AS did not add RVFs to the data packets simply forward the traffic as best-effort.
Although no end-to-end guarantees are achieved this way, such \emph{partial reservations} are useful to bypass ASes that do not participate in \system{}, or to save costs by only protecting traffic on congested links.

\paragraph{Self-Renewing Reservations}
In \system{}, a reservation setup request only contains the authenticated source AS identifier and the reservation directionality flags.
This information is also present in the reservation data packets, and therefore this type of traffic can in principle also serve as a means to renew a reservation (by computing the allowed bandwidth and expiration time, and updating the bandwidth monitor accordingly), even though those packets do not contain an explicit setup request.
This idea leads to the concept of \emph{self-renewing reservations}, i.e., reservations that are not renewed through explicit requests, but implicitly by the source AS sending traffic over the reservation.
Such implicit renewals do not trigger a response packet, and therefore a reservation renewal request is still necessary in case the source AS wants to learn the precise reservation bandwidth.
The reservation expires when the source AS stops generating reservation traffic for a certain time period.
Self-renewing reservations are enabled by \system{}'s ability to efficiently admit bandwidth reservations at border routers.

\subsection{Reservation Billing} \label{sec_billing}
Billing of inter-domain services is a complex topic, and an in-depth analysis is outside the scope of the paper. Nevertheless, we briefly sketch a promising approach that enables simple payments between ASes, and thus supports novel business models based on flyover offerings.
To this end, we first observe that a source should be billed a flat rate for a flyover reservation---regardless of the resulting reservation size---because (i) sources have no control over the resulting amount of bandwidth, and (ii) the additional cost of forwarding flyover- instead best-effort traffic is small, and anyways marginal compared to the overall fixed infrastructure costs~\cite{bw_economics}.
Then, billing the reservation admission can be reduced to billing for the key exchange, as the key is not only needed to secure the reservation setup, but is also necessary to issue reservation requests in the first place. 
Further, the key exchange happens in the control plane, which is not particularly time critical, allowing for additional billing operations.

\section{Related Work} \label{sec_related}

\subsection{GLWP} \label{sec_related_GLWP}
The GMA-based light-weight communication protocol (GLWP)
enables the establishment of inter-domain bandwidth reservations without
on-path ASes having to keep any per-reservation state at their reservation
services. GLWP achieves this by computing a bandwidth reservation for some path
solely based on the traffic matrices of the on-path ASes. The GMA~\cite{GMA}
algorithm, which forms the theoretical foundation behind this computation,
ensures that despite every guarantee being greater than zero, no
over-allocation will ever occur. Intuitively, GMA achieves this property by assigning larger
bandwidths to shorter paths, and (exponentially) smaller bandwidths to longer
paths. GLWP can also support traffic matrices that change over time, without
violating the no-overuse property.
Not keeping state about the per-path reservations comes at a cost, however.
GLWP requires mutual authentication between every pair of on-path ASes to
protect against adversaries on the links that may tamper with the request packet
in order to get a larger reservation or to cause over-allocation. This
quadratic number of authentications adds a large amount of additional state to
the reservation services. 
Further, upon receipt of a reservation request, a GLWP reservation service
needs to verify message authentication codes of all on-path ASes. 
Finally, GLWP assumes a weaker threat model, where no two colluding malicious ASes are neighbors.

\subsection{Colibri} \label{sec_related_Colibri}
Based on the SCION~\cite{SCIONBookv2} network architecture, Colibri establishes inter-domain reservations by assembling multiple shorter reservations, called path segments.
Colibri establishes an end-to-end reservation (EER) over up to three path-segment reservations (SegRs).
In contrast to GLWP, Colibri needs to keep state at each AS about all granted SegRs and EERs in order to keep track of the remaining bandwidth for future reservations. In case of insufficient bandwidth, a Colibri reservation request fails.
Due to its complexity, the algorithm that determines the size of the bandwidth to be admitted for each reservation request is specified and analyzed separately~\cite{ntube}.
A detailed comparison of GLWP and Colibri to \system{} is provided in \cref{sec_appendix_comparison}.

\subsection{Other Bandwidth Reservation Systems} \label{sec_other_reservation_systems}
CoDef~\cite{CoDef} uses a reactive mechanism requesting source ASes to reduce their bandwidth when congestion occurs, which is insufficient to provide strict communication guarantees.
Portcullis~\cite{Portcullis2007} protects packets based on a proof-of-work scheme. While this approach is computationally expensive for normal data traffic, it might be suitable for certain low-rate critical applications. However, Portcullis can only provide probabilistic forwarding guarantees.
The recursive congestion shares (RCS)~\cite{Shenker2020} architecture also provides improved delivery guarantees in the Internet. Nonetheless, the end-to-end resulting allocations are exponentially decreasing in the path length.
STRIDE~\cite{STRIDE} only protects parts of the reservation path, assuming that no congestion arises in the unprotected sub-path.
To the best of our knowledge, only SIBRA~\cite{Basescu2016SIBRA}, GLWP~\cite{GLWP}, and Colibri~\cite{Colibri} satisfy our requirements to a \emph{secure inter-domain} bandwidth reservation system.
As Colibri is the successor of SIBRA, we omit a discussion of SIBRA in this section.

\section{Conclusions \& Future Work}

The public Internet lacks means to protect critical application traffic in a secure and scalable manner.
We address this issue through flyovers: exclusive bandwidth reservations that cross an autonomous system (AS) from ingress to egress border router.
By assembling multiple flyovers, source ASes can create end-to-end protected paths to provide strong delivery guarantees for their critical traffic.
Our simulations based on a topology of \num{5000} ASes show that even at times of maximum demand, every AS can communicate with any other AS using reserved-bandwidth tunnels of over \SI{10}{Mbps}.

We concretely instantiate flyovers in \system{}, a system to protect the flyover setup and forwarding from request flooding and Sybil attacks, and prevent adversaries from overusing their reservation or framing honest ASes.
Our prototype authenticates and forwards reservation traffic at \SI{150}{Gbps} using four cores
on commodity hardware, where a single core achieves a forwarding rate of \SI{6.45}{Mpps}. Even reservation requests can be handled at line rate on the fast path. 
Beyond these performance improvements, the innovative use of flyover reservations enables many advancements over the state of the art, including precise reservation monitoring and policing, flexible traffic control at the reservation source, and an incremental path to deployment for \system{}. 

With this paper we present a step forward in the direction of secure and highly reliable traffic forwarding in the Internet. 
Albeit our prototype is already viable, we see opportunities for future work, such as fleshing out the economic aspects of flyovers or specifying the forwarding behaviour at the source.

\begin{acks}
We would like to thank Carlo Saladin for the interesting discussions regarding the concept of flyover reservations; Patrick Gruntz, for contributing to the correctness proofs; Markus Legner, for his early feedback and suggestions on the system design; Elham Ehsani Moghadam, Simon Scherrer, and Juan Angel García-Pardo for the interesting discussions; and the anonymous reviewers for their helpful feedback.
We gratefully acknowledge support from ETH Zurich, the Zurich Information Security and Privacy Center (ZISC), and Armasuisse.
\end{acks}

\bibliographystyle{IEEEtran}
\bibliography{bib/ref}
\appendix

\section{Other Reservation Computations}\label{sec_appendix_alternatives_to_pip_reservations}
At first glance, two alternative approaches to per--interface-pair reservations suggest themselves:
per--ingress-interface reservations, where the bandwidth of the ingress link is shared ($\sum_j \matentry{k}{i}{j}$), and per--egress-interface reservations, where the bandwidth of the egress link ($\sum_i \matentry{k}{i}{j}$) is shared among the requesting ASes.
While those approaches increase the flexibility of a source AS, they can not provide any communication guarantees.
With per--ingress-interface reservations, if many ASes decide to send traffic over multiple ingress interfaces towards the same egress interface of a shared provider AS, congestion can arise and packets will unavoidably be dropped. This becomes particularly evident in case the egress link has low capacity.
The same reasoning applies to per--egress-interface reservations, where congestion occurs for example if there is an egress interface with a high capacity link and an ingress interface with comparably low capacity.
Also the combination of per-ingress- and per--egress-interface reservations is insufficient, as traffic may already get dropped in the intra-AS network.

\section{Supported Network Architectures} \label{sec_appendix_deployment_flyovers}
In our network model (\cref{sec_network_model}), we assume that the Internet architecture provides paths that are stable over the reservation validity periods.
Without path stability, flyover reservations would be invalidated by path changes, therefore voiding the possibility of achieving delivery guarantees.
Path stability is also explicitly required by other bandwidth reservation systems (\cref{sec_related}).
Further, to make use of \system{}'s bidirectional reservations, we also need to assume that paths are symmetric, i.e., that packets between some source and destination traverse the same sequence of interfaces in both directions.
In the following, we discuss path stability and path symmetry in the current Internet as well as in next-generation network architectures.

\paragraph{BGP}
In the current Internet, based on the Border Gateway Protocol (BGP), path stability can generally not be guaranteed.
While paths have been shown to be stable on the order of days~\cite{path_stability}, there is no absolute guarantee that paths do not change unpredictably: routing events can trigger network re-convergence and hijacking attacks can redirect traffic to adversary-controlled paths.
Moreover, the majority of BGP paths are asymmetric~\cite{path_symmetry}.

Still, there are several means to achieve path stability and path symmetry in the current Internet.
For instance, the ARROW system~\cite{arrow}---originally designed to protect against routing attacks---achieves these properties by having ISPs offer tunneled transit through their networks. Customers can then buy and combine such transits from one or multiple ISPs.
As a more general framework, Platypus~\cite{platypus2004,platypus2009} achieves path stability through source routing based on capabilities, where end users can select the paths their traffic takes. Path symmetry in Platypus is implemented by means of reply capabilities.

\paragraph{Next-Generation Architectures}
To overcome BGP's shortcomings, several next-generation Internet architectures have been proposed, many of them providing path stability and path symmetry by design.
SCION~\cite{SCIONBookv2}, for example, achieves path stability by decoupling routing decisions from the dissemination of path information, thus avoiding any path convergence issues. SCION paths are encoded in form of packet-carried forwarding state (PCFS) at the level of AS interfaces, providing optimal conditions for \system{}'s deployment. Paths are valid on the order of several hours and are guaranteed not to change during that period, making them resistant to hijacking attempts.
NEBULA~\cite{NEBULA} and XIA~\cite{XIA} are further examples of network architectures providing path stability.

\section{Computing the Number of Requesting ASes} \label{sec_appendix_fine_grained}
We provide an algorithm for the flyover size computation that estimates $\rho$, i.e., the number of requesting ASes, on each border router.
The algorithm consists of two subroutines: \cref{alg_fine_grained} calculates the reservation size and is hence executed at the border router every time it receives a reservation request, while \cref{alg_update} is called periodically at intervals of~$\epsilon$ and serves the purpose of updating $\rho$ based on the demand observed in the near past.

To keep track of the requesting ASes, our algorithm needs three Bloom filters: \bloom{P}, \bloom{C}, and \bloom{CC}. \bloom{P} contains the identifiers of all ASes that can successfully be provided with a flyover reservation in the current interval. \bloom{C} collects the ASes that issued a request in the current interval, and \bloom{CC} is equal to the state of \bloom{C} at the end of the previous interval.
Additionally, the variable $\rho$ is used to compute the reservation size for requests arriving in the current interval and $\rho_{min} \geq 1$ denotes the minimum value that $\rho$ can assume.
Lastly, $\omega (0 < \omega \leq 1)$ denotes the maximum portion of allocation matrix bandwidth that is allowed to be allocated by the algorithm.

In \cref{alg_update}, when a new interval starts, the variables and bloom filters are updated:
$\rho$ is assigned the number of ASes that issued a request in the last interval ($|\bloom{C}|$) or in the penultimate interval ($|\bloom{CC}|$). The algorithm ensures that $\rho$ is always at least as large as $\rho_{min}$.
In the following step, the Bloom filters are swapped: \bloom{CC} becomes \bloom{P}, \bloom{C} becomes \bloom{CC}, and \bloom{P} is reset (all ASes are removed, i.e., the Bloom filter bits set to zero) and becomes \bloom{C}.

In \cref{alg_fine_grained}, upon arrival of a reservation request, the requesting AS~$S$ is added to \bloom{C}. If AS~$S$ is already in \bloom{P}, it receives a flyover reservation of size $\frac{\omega\cdot\matentry{k}{i}{j}}{\rho}$, that will expire after an interval of $\epsilon$.
If AS~$S$ is not in \bloom{P}, the reservation request fails.
In this case, AS~$S$ can still issue requests in the next intervals.
As shown in \cref{sec_appendix_proofs}, at most $2\epsilon$ after its first request, AS~$S$ is guaranteed to receive a flyover reservation.
Because the algorithm provides at most $\omega\cdot\matentry{k}{i}{j}$ bandwidth to all requesting ASes together, there is $(1-\omega)\cdot\matentry{k}{i}{j}$ bandwidth left.
The border router can divide this leftover bandwidth among ASes for which the bandwidth admission in \cref{alg_fine_grained} failed, so that they still get a reservation already at their first request. For instance, the border router can keep $\theta$ fixed slots with bandwidth $\frac{(1-\omega)\cdot\matentry{k}{i}{j}}{\theta}$, and hence in every interval $\epsilon$ provide up to $\theta$ ASes with such tentative reservations.

This algorithm can be deployed either once on an ingress border router, or multiple times for connections to different egress interfaces. While the first approach is more frugal in terms of memory overhead, the latter approach allows for more fine-grained estimates of $\rho$ and can hence provide higher reservation sizes.
Also, Bloom filter sizes can be chosen on each router separately based on the topology and historical data.
For example, if a border router expects \num{10000} ASes, it could choose \num{7} hash functions for the Bloom filter; to achieve a false positive probability of \num{1} percent, the memory overhead required by the Bloom filter would then be less than \num{12}{KB}~\cite{bloomfilter}.
Furthermore, we suggest to choose $\epsilon$ in the order of ten seconds, and to set $\omega$ around \num{1.2}, meaning that apart from the $|\bloom{P}|$ registered ASes, additional $0.2\times |\bloom{P}|$ unregistered ASes can immediately get a flyover reservation on their first request.
The update times, i.e., the times at which a new interval starts, do not need to be synchronized between different border routers or ASes.

\begin{algorithm}[t]
\DontPrintSemicolon
  % Input
  \KwInput{\bloom{P}, \bloom{C}, $\rho$, $\omega$, $S$, $i$, $j$}

  % Output
  \KwOutput{$\bw{}$, \tsexp{}}

  % Algorithm
  $\bloom{C} \leftarrow \bloom{C}~\cup~\{S\}$ \\
  \If{$S \in \bloom{P}$}
  {
    \Return{$\frac{\omega\cdot\matentry{k}{i}{j}}{\rho},~\text{now()}+\epsilon$}
  }
  \Else
  {
    \Return{$\bot, \bot$}
  }

\caption{Fine-grained bandwidth admission. Executed upon arrival of a reservation request at AS~k.}
\label{alg_fine_grained}
\end{algorithm}

\begin{algorithm}[t]
\DontPrintSemicolon

  % Input
  \KwInput{\bloom{P}, \bloom{C}, \bloom{CC}, $\rho$}
  
  % Algorithm
  $\rho \leftarrow \text{max}(|\bloom{C}\cup\bloom{CC}|,~\rho_{min})$ \\
  $(\bloom{P},~\bloom{CC},~\bloom{C}) \leftarrow (\bloom{CC},~\bloom{C},~\bloom{P}.\text{Reset()})$

\caption{Update procedure. Executed after a time~$\epsilon$ following the last execution.}
\label{alg_update}
\end{algorithm}

\section{Correctness proofs} \label{sec_appendix_proofs}
\allowdisplaybreaks

\begin{theorem} \label{th_two_epsilon}
At most $2\epsilon$ after the first request, the requesting AS will successfully be granted a flyover reservation.
\end{theorem}
\begin{proof}
Assume the AS requests a reservation at some time~$t$. The AS is then added to \bloom{C}.
At time $t+\epsilon$, the next interval will have started, and the AS will be in \bloom{CC} due to the Bloom filter rotation at the end of \cref{alg_update}.
If an AS sends a request at $t+2\epsilon$, the interval will have changed again so that the AS will be in \bloom{P}, and hence \cref{alg_fine_grained} will return a flyover reservation.
\end{proof}

\begin{theorem} \label{the_no_overallocation}
Assume that there are no false positives occurring at any of the Bloom filters. At every point in time, the sum of all granted flyover reservations for the interface-pair (i, j) of AS~k is smaller than $\omega\cdot\matentry{k}{i}{j}$.
\end{theorem}
\begin{proof}

We formulate the algorithm in a form that is simpler to analyze.
For that, consider four arbitrary but consecutive intervals, where $R_1, R_2, R_3$, and $R_4$ denote the set of ASes that issue a request in the corresponding interval.
We want to prove that there is no over-allocation in the fourth interval.
For this, we need to show that the sum of the flyovers granted to the successful ASes in $R_3$ plus the bandwidth granted to the successful ASes in $R_4$ is less than $\omega\cdot\matentry{k}{i}{j}$.
Note that an AS in $R_3$ is only successful in receiving a flyover reservation if it also issued a request in the first interval, i.e., if the AS is in $R_1$. If it is successful, the AS gets a flyover reservation of $\frac{\omega\cdot\matentry{k}{i}{j}}{|R_1\cup R_2|}$ starting in the third interval.
Similarly, an AS in $R_4$ is only successful if it is also in $R_2$. If successful, it gets a bandwidth of $\frac{\omega\cdot\matentry{k}{i}{j}}{|R_2\cup R_3|}$ starting in the fourth interval.
Also, an AS that successfully got a reservation both in the third and fourth interval, i.e., it is in $R_1\cap R_2\cap R3\cap R4$, will only be able to use one of those reservations at times where they overlap.
Let $B$ be the total bandwidth provided as flyover reservations in the fourth interval. Then we have:
\begin{align*}
    B &\leq |R_1\cap R_2\cap R3\cap R4|\cdot\max \big\{\frac{\omega\cdot\matentry{k}{i}{j}}{|R_1\cup R_2|}, \frac{\omega\cdot\matentry{k}{i}{j}}{|R_2\cup R_3|}\big\} \\
    &~~~~+ |(R_1\cap R_3)\backslash(R_2\cap R_4)|\cdot \frac{\omega\cdot\matentry{k}{i}{j}}{|R_1\cup R_2|} \\
    &~~~~+ |(R_2\cap R_4)\backslash(R_1\cap R_3)|\cdot \frac{\omega\cdot\matentry{k}{i}{j}}{|R_2\cup R_3|}
\end{align*}
Case $\frac{\omega\cdot\matentry{k}{i}{j}}{|R_1\cup R_2|} \geq \frac{\omega\cdot\matentry{k}{i}{j}}{|R_2\cup R_3|}$:
\begin{align*}
    B &\leq |(R_1\cap R_3)\cap (R2\cap R4)|\cdot \frac{\omega\cdot\matentry{k}{i}{j}}{|R_1\cup R_2|} \\
    &~~~~+ |(R_1\cap R_3)\backslash(R_2\cap R_4)|\cdot \frac{\omega\cdot\matentry{k}{i}{j}}{|R_1\cup R_2|} \\
    &~~~~+ |(R_2\cap R_4)\backslash(R_1\cap R_3)|\cdot \frac{\omega\cdot\matentry{k}{i}{j}}{|R_1\cup R_2|} \\
    &= \frac{|(R_1\cap R_3)\cup (R_2\cap R_4)|}{|R_1\cup R_2|} \cdot \omega\cdot\matentry{k}{i}{j} \\
    &\leq \frac{|R_1\cup R_2|}{|R_1\cup R_2|} \cdot \omega\cdot\matentry{k}{i}{j} \\
    &= \omega\cdot\matentry{k}{i}{j} \\
\end{align*}
Case $\frac{\omega\cdot\matentry{k}{i}{j}}{|R_1\cup R_2|} < \frac{\omega\cdot\matentry{k}{i}{j}}{|R_2\cup R_3|}$:
\begin{align*}
    B &\leq |(R_1\cap R_3)\cap (R2\cap R4)|\cdot \frac{\omega\cdot\matentry{k}{i}{j}}{|R_2\cup R_3|} \\
    &~~~~+ |(R_1\cap R_3)\backslash(R_2\cap R_4)|\cdot \frac{\omega\cdot\matentry{k}{i}{j}}{|R_2\cup R_3|} \\
    &~~~~+ |(R_2\cap R_4)\backslash(R_1\cap R_3)|\cdot \frac{\omega\cdot\matentry{k}{i}{j}}{|R_2\cup R_3|} \\
    &= \frac{|(R_1\cap R_3)\cup (R_2\cap R_4)|}{|R_2\cup R_3|} \cdot \omega\cdot\matentry{k}{i}{j} \\
    &\leq \frac{|R_2\cup R_3|}{|R_2\cup R_3|} \cdot \omega\cdot\matentry{k}{i}{j} \\
    &= \omega\cdot\matentry{k}{i}{j} \\
\end{align*}
The use of $\rho_{min}$ in the algorithm ensures that $|R_1\cup R_2|$ and $|R_2\cup R_3|$ are both greater than zero.
\end{proof}

\section{Token-Bucket without Counter} \label{sec_appendix_token_bucket}
In standard literature, the token-bucket algorithm is either implemented using a counter and a timer or using a counter and a timestamp~\cite{networkrouting2007}. The first approach does not scale as for every flow (or every source AS, in the case of \system{}) a separate timer is needed. Also, timers are often not suitable for high-performance applications. The second approach is computationally already very efficient---we show through \cref{alg_token_bucket} that it is possible to further improve the implementation by only keeping a single timestamp, thus achieving better memory efficiency.

\begin{algorithm}[t]
\DontPrintSemicolon
  \KwInput{CIR, T, pkt, ts}
  pkt-len $\leftarrow$ len(pkt)\\
  now $\leftarrow$ time()\\
  pkt-time $\leftarrow$ pkt-len $\div$ CIR\\
  \If{\upshape max(ts, now) + pkt-time $\leq$ now + T}
    {
    ts $\leftarrow$ max(ts, now) + pkt-time\\
    \Return{\upshape true}
    }
  \Else
  {
    \Return{\upshape false}
  }
\caption{Memory-efficient token-bucket implementation. The configuration happens through a particular choice of the committed information rate (CIR), committed burst size (CBS), and the time interval (T), where they are related through $\text{T}\times\text{CIR} = \text{CBS}$. The token-bucket consists only of a timestamp (ts).}
\label{alg_token_bucket}
\end{algorithm}

\section{Comparison} \label{sec_appendix_comparison}
This section highlights the most important differences between GLWP, Colibri, and \system{}.

\subsection{Flexibility}
\paragraph{Traffic Control}
In GLWP and Colibri, AS-level reservations are almost static.
The size of a reservation in GLWP is always the same, even after renewal, where allocation matrix changes at on-path ASes constitute the only exception.
In Colibri, SegRs have a validity of approximately five minutes.
Therefore, if two reservation paths overlap, and during some time interval one SegR is fully utilized while the other is not, a source AS in Colibri cannot quickly move free capacity from one reservation to the other.
While even inherently impossible in GLWP, dynamically moving reservation capacity between overlapping paths can be done \emph{instantaneously} through traffic control in \system{} (see \cref{sec_traffic_control}).

\paragraph{Dynamic Bandwidths}
In Colibri, an AS specifies the bandwidth it wants to reserve in the reservation request.
Reserving only little bandwidth means that other ASes can reserve more.
GLWP and \system{} by default do not provide this flexibility.
A concrete design of a demand-aware flyover version is presented in \cref{sec_appendix_demand_aware}.

\paragraph{Allocation Matrix Changes}
In Colibri, traffic matrices are assumed to be static---the possibility to change them over time is not discussed in the literature.
In principle this is feasible however, by always computing with the more conservative values from the old and new allocation matrix during reservation setups and renewals. Decreasing an entry in the traffic matrix takes at least as long as the maximum validity period of the SegRs, which is in the order of minutes.
In GLWP reservations can stay valid for even longer time periods, and thus updating the traffic matrix also takes longer.
Moreover, increasing an allocation matrix entry in GLWP can cause reservations through seemingly unaffected interface-pairs to actually decrease.
All these issues do not occur in \system{}, as the duration of an allocation matrix update is in the order of seconds.

\paragraph{Reservation Service Recovery}
For stateful reservation systems such as Colibri, a collapsing reservation service has far-reaching consequences.
Loosing information about the authenticators means that existing reservations become unusable.
An AS can try to create new reservations, but those may be small or even get denied, because the existing but inaccessible reservations take up large portions of the link capacities.
Existing reservations for which an AS acts as transit provider are unaffected by a service failure, but they cannot be renewed, and no further reservations can be established.
The latter is due to the service not having information about the current reservations anymore, and hence it does not know the fraction of free network capacity.
In contrast, lost authenticators in GLWP and \system{} can be requested again in a single round trip.
Moreover, the service in GLWP is stateless with respect to transit reservations, and in \system{} the setup and renewal packets do not even need to go through a reservation service at on-path ASes. 

\paragraph{Incremental Deployment}
Whereas in GLWP and Colibri every AS on a reservation path is required to support the corresponding protocol, \system{} can be deployed incrementally: the header fields of the reservation setup and data transmission packets are ignored by ASes not participating in the protocol.
The forwarding and QoS guarantees increase with every AS deploying \system{}.

\subsection{Performance}\label{sec_appendix_comparison_performance}

\paragraph{Reservation Admission}
The bandwidth computation for a reservation request in \system{}, as specified in \cref{alg_fine_grained}, only takes \SI{30}{ns}.
The overall admission overhead including authentication and encryption operations amounts to \SI{618}{ns}, which means that admission can be performed at line rate.
In GLWP, the admission overhead strongly depends on the number of ASes along the reservation path.
Both the bandwidth computation according to the GMA algorithm as well as the number of MAC calculations increase with every additional on-path AS.
The inputs to the MACs are longer than in \system{}, and therefore more block cipher operations are needed in their computations.
Colibri requires \SI{0.4}{ms} and \SI{1.25}{ms} to process an EER- and SegR admission, respectively, which does not yet account for the authentication overhead~\cite{Colibri}.
Opposed to \system{}, ASes in GLWP and Colibri need to process the setup request not only once in the forward direction but also again when the packet is returned from the destination to the source.

\paragraph{Transport Layer}
In Colibri, a setup request is computationally much more expensive compared to GLWP and \system{}, which makes packet loss along the desired reservation path particularly undesirable.
While in GLWP and \system{} the source AS can simply re-issue a request packet after a timeout, Colibri needs a reliable transport protocol such as TCP \cite{TCP} or QUIC \cite{QUIC} between reservation services of neighboring ASes.
This adds complexity and requires additional state at the reservation services, which is important to consider given the frequent EER setup and renewal requests.

\paragraph{Communication Overhead}
Ideally, the header of a reservation packet is as small as possible, since a larger header consumes bandwidth that could otherwise be used for data traffic. 
Assuming all header fields that serve the same purpose also have the same length in GLWP, Colibri, and \system{}, e.g., the source AS identifier, the potentially included reservation path, the validation fields or the packet timestamp, then for unidirectional reservation traffic the Colibri header is \num{17} and the GLWP \SI{15}{\byte} longer than the \system{} header.\footnote{Assuming sizes for the remaining header fields as follows. GLWP: $\beta$~(\SI{4}{B}), AID~(\SI{8}{B}), tsExp~(\SI{4}{B}). Colibri: ResID~(\SI{12}{B}), Bw~(\SI{1}{B}), ExpT~(\SI{4}{B}), Ver~(\SI{1}{B}). \system{}: D~(\SI{1}{B}). Note that for unidirectional traffic the lenB~(\SI{2}{B})  field does not need to be included in \system{} packets.}
Hence for an average Internet path length of five AS-level hops, \system{} can support bidirectional reservations, i.e., include its BVFs of \SI{3}{\byte} each, to reach around the same header size as GLWP or Colibri require for unidirectional traffic.

\paragraph{Authenticator Updates}
Reservation services in Colibri and GLWP need to update their authenticators every time they renew a reservation.
In contrast, authenticator updates are rare events in \system{}, they only occur when the keys established through DRKey change.
In GLWP and \system{}, multiple critical applications share the same reservation, i.e., the same authenticators, and the reservation service in the source AS is responsible to fairly distribute the available reservation bandwidth among those applications.
In Colibri, however, every endpoint requests and then frequently renews the authenticators for its own EERs.

\paragraph{Lack of Available Bandwidth}
As bandwidth sizes in GLWP decrease exponentially with the length of a reservation path, they can become too small to be practically useful, in which case a reservation setup request fails.
In Colibri, reservation requests are denied if there is not enough free capacity along the reservation path.
In \system{}, a request is only rejected if a many ASes want to suddenly establish a reservation for the very first time at the same inter-face pair of the same AS (\cref{sec_appendix_fine_grained}).
Also, a request sent $2\epsilon$ after the initial request is guaranteed to succeed (\cref{th_two_epsilon}).

\paragraph{Reservation Granularity}
In Colibri, every end host has its own reservation, which not only implies more overhead at the end hosts due to reservation management, but also at ASes due to fine-grained reservation monitoring and policing, frequent reservation admissions and communications, and for keeping state.
In GLWP and \system{}, the reservation service of the source AS is responsible to fairly distribute reservation bandwidth among end hosts. This results in a much reduced overhead as compared to Colibri.

\subsection{Security}

\paragraph{Commonalities}
All three systems have been designed with high demands on their security properties.
In particular, they all provide botnet-size independence, the property that a legitimate inter-domain reservation cannot be indefinitely reduced by malicious ASes requesting a high number of reservations or excessive reservation bandwidth.
Colibri, GLWP, and \system{} all rely on a bandwidth monitor, a replay-suppression system, and a framework to efficiently establish shared symmetric keys.
They also have similar assumptions about the underlying network architecture, e.g., path stability and path symmetry, and time synchronization.

\paragraph{Assumptions and Requirements}
In addition to the assumptions and requirements shared by all three protocols, GLWP needs a dedicated sub-protocol to ensure the consistency of values promoted by neighboring ASes.
Also, GLWP assumes that no two adjacent ASes are malicious, and that symmetric keys are already established between every pair of ASes.
The latter assumption is necessary, because every on-path AS authenticates every other on-path AS during the reservation admission.

\paragraph{Malicious Parameter Announcement}
Colibri and GLWP need to have checks in place to detect ASes that announce malicious parameters (e.g., wrong information about their traffic matrices) contributing to the computation of the path-based reservation bandwidth.
In \system{} no such measures are necessary, as the only information an on-path AS needs to compute the bandwidth is the authenticated source AS identifier.
In particular, this computation does not depend on any values claimed by other ASes.

\paragraph{Monitoring}
As GLWP and Colibri are path-based reservation protocols, ASes need to rely on a probabilistic bandwidth monitor to police transit traffic of the (potentially arbitrarily) many reservation paths that cross their network.
As the number of ASes in the Internet is more stable and most notably way lower than the possible number of reservation paths, ASes in \system{} can deterministically monitor all reservation traffic.
False positives, i.e., misclassifications of non-overuse reservations as overuse reservations, and false negatives, i.e., misclassifications of overuse reservations as non-overuse reservations, are therefore inherently impossible in \system{}.

\paragraph{DDoS against Reservation Services}
In GLWP and Colibri, the reservation services constitute a target for DDoS attacks that attempt to overwhelm the admission procedure, which would cause the service to drop legitimate requests.
In contrast, bandwidth admission in \system{} happens at line rate, and hence a malicious request flood does not have an influence on honest requests--no honest request is ever dropped due to lack of computational resources.

\paragraph{Suitability for DoCile}
DoCile~\cite{docile} is a framework for the protection of key exchanges and reservation setup requests from denial-of-capability (DoC) attacks.
Consider a source AS~$S$ that wants to fetch symmetric keys from all ASes on a certain path~$P$, in order to subsequently establish a reservation over that path.
DoCile protects this communication by having AS~$S$ iteratively fetch keys and establish a reservation on every sub-path of $P$.
First, a key is requested from the first on-path AS~$O_1$, which is used to establish a reservation for the path [$S$, $O_1$].
Then, protected through this reservation, AS~$S$ requests a key from the second on-path AS~$O_2$ and subsequently establishes a reservation over the path [$S$, $O_1$, $O_2$].
This process is repeated until ultimately a reservation over $P$ is established.
Implementing DoCile in GLWP or Colibri would be inefficient, as a new reservation has to be created for each sub-path.
In \system{}, however, a flyover reservation can be reused in multiple DoCile iterations.
To provide the highest guarantees possible, DoCile assumes that the reservation system in use provides bidirectional reservations. Bidirectional reservations are only available in \system{}, but not in GLWP and Colibri.

\subsection{Complexity}
The flyover computation in \system{} is based on the idea that bandwidth between two interfaces is split equally among the requesting ASes.
Desirable properties such as non-zero bandwidth sizes or no-over-allocation directly follow from this computation.
In GLWP, the bandwidth size computation is based on the GMA~\cite{GMA} algorithm, and Colibri relies on N-Tube~\cite{ntube} for this purpose.
In contrast to flyovers, GMA and N-Tube rely on a theoretical foundation that is quite complex.
The simplicity of the flyover calculation is the reason that \system{} can achieve its high admission and traffic forwarding rates, and that our bandwidth monitoring system can accomplish 100\% precision.
We consider low complexity important for reasoning about \system{}'s security and to achieve wide-spread adoption in a real-world deployment.

\section{Demand-Aware Flyovers} \label{sec_appendix_demand_aware}
The flyover bandwidth computation algorithm presented in \cref{eq_bandwidth} equally shares the available capacity among all requesting ASes.
For an AS requesting a flyover it is thus not possible to express its bandwidth demands.
This demand-agnostic design enables highly efficient bandwidth admission directly at the border router, while still providing sufficiently high reservation bandwidths to run the critical applications targeted by our system.
However, flyover sizes can also be computed in other ways, in particular to also take into account bandwidth demands.
We designed \system{} to be independent of the concrete flyover computation: the function \texttt{getBandwidth()} used in \cref{alg_admission} can be instantiated with a custom algorithm.

\newcommand{\bwdem}[0]{\bw{\text{dem}}}
\newcommand{\bwmin}[0]{\text{\bw{min}}}
Concretely, to achieve demand-aware flyover reservations, one might add an authenticated bandwidth demand \bwdem (as well as the minimal acceptable bandwidth \bwmin{}) in the setup request (\cref{sec_reservation_setup}), and substitute \texttt{getBandwidth()} with, e.g., the demand-aware N-Tube algorithm that is also used in Colibri.
The setup request would therefore need to be adapted as follows:
\begin{align*}
    \text{Auth}_i &= \mac{\drkey{i}{S}}{\text{tsReq}, R_i, B_i, \bwdem{}, \bwmin{}} \\
    \textsc{SetupReq} &= S, \text{tsReq}, \bwdem{}, \bwmin{}, (R_i, B_i, \text{Auth}_i) \\
    \forall i & \in A \subseteq \{1, \dots, \ell\}
\end{align*}
Other packet fields or computations, e.g., the calculation of the authenticator (\cref{eq_auth}) or the RVF (\cref{eq_rvf}) do not need to be modified.
The N-Tube algorithm readily computes AS-local interface-pair reservations, which Colibri further assigns to specific paths in the form of segment reservations.
The computation of interface-pair reservations makes N-Tube a viable flyover algorithm.
Its fairness, security, and performance properties have already been extensively evaluated~\cite{Colibri,ntube}.
Due to its comparably large performance and state overhead, it may need to be deployed in a dedicated service inside the AS instead of being executed at the border router.
However, in contrast to Colibri, no segment and end-to-end reservations need to be computed and stored.

Note that even demand-aware protocols such as N-Tube implicitly become demand-agnostic when the aggregate demand exceeds the available capacity: Shrinking existing reservations to a fair share is necessary to allow all parties to communicate.
Finally, one of Helia's distinguishing features is that every AS can deploy its own flyover algorithm (possibly even at the granularity of single interface-pairs), it does not need be globally agreed upon.
This additional \emph{algorithmic} agility further eases deployment and favors diverse business cases.

\end{document}